\pdfoutput=1
\documentclass{article}

\usepackage[utf8]{inputenc} 
\usepackage{multirow}
\usepackage{amsmath,amssymb,amsthm}	
\usepackage{dsfont}
\usepackage{mathrsfs}
\usepackage{mathtools}
\usepackage{bbm}
\usepackage{ifthen}
\usepackage{graphicx}
\usepackage{thm-restate}
\usepackage{subcaption}
\usepackage{booktabs}
\graphicspath{ {./images/} }
\usepackage{hyperref}			
\usepackage{color}
\usepackage[dvipsnames]{xcolor}
\usepackage[
backend=biber,
style=apa,natbib=true
]{biblatex}

\usepackage{enumitem}
\usepackage[ruled,vlined]{algorithm2e}
\usepackage{hyperref}
\usepackage{lscape}
\usepackage{parskip}

\usepackage{todonotes}

\theoremstyle{plain}
\newtheorem{lemma}{Lemma}

\theoremstyle{plain}
\newtheorem{theorem}{Theorem}

\theoremstyle{plain}

\theoremstyle{plain}

\theoremstyle{plain}

\theoremstyle{plain}

\newtheorem{definition}{Definition}[section]

\newtheorem{corollary}{Corollary}

\newtheorem{remark}{Remark}

\theoremstyle{remark}

\theoremstyle{remark}

\newcommand{\floor}[1]{\lfloor #1 \rfloor}

\newcommand{\abs}[1]{|#1|}
\newcommand{\card}[1]{\left \lvert#1 \right \rvert}

\newcommand{\Ou}{\mathcal{O}}
\newcommand{\Prix}{\mathcal{P}}
\newcommand{\bvec}{\mathbf{b}}
\newcommand{\hvec}{\mathbf{h}}
\newcommand{\betavec}{\boldsymbol{\beta}}
\newcommand{\Observed}{A}

\newcommand{\Expectation}[1][]{ 
    \ifthenelse{ \equal{#1}{} }
    {\mathbb{E}}
    {\mathbb{E} \left [ #1 \right ] }
}
\newcommand{\Proba}[1][]{ 
    \ifthenelse{ \equal{#1}{} }
    {\mathbb{P} }
    {\mathbb{P} \left ( #1 \right )}
}
\newcommand{\Probat}[1][]{ 
    \ifthenelse{ \equal{#1}{} }
    {\mathbb{P}^t }
    {\mathbb{P}^t \left ( #1 \right )}
}
\newcommand{\HalfProba}[1][]{ 
    \ifthenelse{ \equal{#1}{} }
    {Q }
    {Q^{#1}}
}
\newcommand{\HalfProbat}[1][]{ 
    \ifthenelse{ \equal{#1}{} }
    {Q^t }
    {Q^t \left ( #1 \right )}
}
\newcommand{\indicator}[1][]{ 
    \ifthenelse{ \equal{#1}{} }
    {\mathds{1} }
    {\mathds{1} \left \{ #1 \right \}}
}
\newcommand{\Bigindicator}[1][]{ 
    \ifthenelse{ \equal{#1}{} }
    {\mathds{1} }
    {\mathds{1} \Big{\{} #1 \Big{\}} }
}

\usepackage{dsfont}
\DeclareRobustCommand{\[}{\left [}

\newcommand\numberthis{\addtocounter{equation}{1}\tag{\theequation}}

\usepackage[nonatbib,final]{neurips_2024}

\usepackage{graphicx} 

\newcommand{\nin}{N^{\mathrm{in}}}

\title{Improved learning rates in multi-unit uniform price auctions}
\author{
  Marius Potfer\thanks{marius.potfer@ensae.fr}$^{\ \ 1,2}$  \\
   \And
  Dorian Baudry$^3$ \\
   \And
  Hugo Richard$^1$ \\
   \And
  Vianney Perchet$^1$ \\
  \And 
  Cheng Wan$^{2}$ \\
}

\begin{document}

\maketitle
\vspace{-1cm}
\begin{center}
\begin{tabular}{l}
   $^1$ Joint team Fairplay, ENSAE, and Criteo AI LAB\\
   $^2$ EDF R\&D \\
   $^3$ Department of Statistics, University of Oxford \\
\end{tabular}
   \vspace{2mm}
\end{center}

\begin{abstract}
Motivated by the strategic participation of electricity producers in electricity day-ahead market, we study the problem of online learning in repeated multi-unit uniform price auctions focusing on the adversarial opposing bid setting. The main contribution of this paper is the introduction of a new modeling of the bid space. Indeed, we prove that a learning algorithm leveraging the structure of this problem achieves a regret of $\tilde{O}(K^{4/3}T^{2/3})$ under bandit feedback, improving over the bound of $\tilde{O}(K^{7/4}T^{3/4})$ previously obtained in the literature. This improved regret rate is tight up to logarithmic terms. 
Inspired by electricity reserve markets, we further introduce a different feedback model under which all winning bids are revealed. This feedback interpolates between the full-information and bandit scenarios depending on the auctions' results. We prove that, under this feedback, the algorithm that we propose achieves regret $\tilde{O}(K^{5/2}\sqrt{T})$. 
\end{abstract}

\section{Introduction} \label{introduction}

The short-term electricity market, based on a wholesale market, is organized as an auction that determines the quantity each electricity producer needs to produce and the price at which electricity is sold. They participate, in this market by submitting prices for each kilowatt-hour they can produce. They have the opportunity to participate strategically, submitting prices that can deviate from their actual production cost. While several regulatory and practical constraints apply, this market is essentially a multi-unit auction of identical items \citep{willems2022bidding}. These auctions are extensively studied and utilized for resource allocation. Several pricing rules can be applied, the most common being discriminatory pricing, uniform pricing \citep{demand_red_ausubel}, and Vickrey–Clarke–Groves (VCG) auctions \citep{sessa2017exploring}. Although the VCG auction is known for its truthful bidding property, it is seldom implemented due to its complexity. Instead, uniform or discriminatory pricing are often preferred, particularly in treasury auctions \citep{khezr2022review} and their procurement variations in electricity reserve markets \citep{viehmann2021multi}. 

The wholesale electricity market is held every day and, structurally, the electricity producers who participate in the mechanism remain the same for multiple years. This represents an opportunity to study how producers can be strategic in the way they adapt to other's bidding strategies. We therefore focus on the problem of online bidding in a repeated multi-unit auction. This setting allows us to model how an agent participating multiple times to an auction with the same other participants can leverage the information he has obtained during past auctions. This family of settings, of which a review is available in \citep{nedelec2022learning}, was first investigated for learning from the point of view of the auctioneer and was then applied to bidders learning how to bid optimally. 
Online learning in multi-unit auctions with uniform pricing (every object is sold at the same price independently of the winner) is studied in \citep{branzei_learning_2024} while the case of discriminatory pricing is studied in~\cite{galgana2023learning}.
When the bids of all bidders are revealed after each auction (full-information), known regret rates for uniform and discriminatory pricing are of the same order $\mathcal{O}(\sqrt{T})$ \citep {branzei_learning_2024, galgana2023learning} where $T$ is the time horizon. When bidders only observe the number of items they win and the price (bandit feedback) the regret upper bounds given in \cite{galgana2023learning} and \cite{branzei_learning_2024} are of order  $\mathcal{\tilde{O}}(T^{2/3})$ (for discriminatory pricing) and $\mathcal{\tilde{O}}(T^{3/4})$ (for uniform pricing) suggesting that bidding multi-unit auctions with uniform pricing is strictly harder than with discriminatory pricing. Our study shows that this is not the case as we present an algorithm achieving regret $\mathcal{\tilde{O}}(T^{2/3})$ with uniform pricing therefore closing the gap between the two settings. 

\paragraph{Auction rules} A decision-maker (i.e., \emph{the bidder}) repeatedly bids in a uniform pricing $K$-unit auction. The single-shot version of the auction, from the perspective of any participant $i$ whose value of obtaining a $\text{k}^\text{th}$-item is denoted by $v_{i,k }\in [0,1]$, proceeds as follows.
\begin{enumerate}
\item Each participant submits a bid profile $(b_{i,k})_{k\in[K]} \in B$, where $$B=\{(b_k)_{k \in [K]}, \text{ such that } 1 \geq b_{1}\geq b_{2} \geq \hdots \geq b_{K} \geq0 \}.$$ We call $B$ the action space and we denote by $\bvec_{-i}$ the bids from other participants.
\item The price per item $p \left ( \bvec_i,\bvec_{-i} \right )$ is set either as\begin{itemize} [parsep=0cm,itemsep=0.05cm,topsep=0cm]
    \item the $K^{\text{th}}$ highest bid (Last Accepted Bid (LAB) pricing rule),
    \item the $(K+1)^{\text{th}}$ highest bid (First Rejected Bid (FRB) pricing rule).
\end{itemize}
\item The participant receives the items they won and pays $p \left ( \bvec_i,\bvec_{-i} \right )$ for each item. Since items are identical, we call allocation $x_i\in[K]:=\{1,2,...,K\}$ the number of items participant $i$ receives, formally defined as follows:
        \begin{align*}
            x_i \left (   \bvec_i,\bvec_{-i}  \right ) & := \left \{ 
                \begin{matrix} 
                    \card{ \left \{k \in [K] \text{ s.t. } b_{i,k} \geq p\left (  \bvec_i,\bvec_{-i}  \right ) \right \} } \text{for the LAB rule} \\[1em]
                    \card{ \left \{k \in [K] \text{ s.t. } b_{i,k} > p\left (  \bvec_i,\bvec_{-i}  \right ) \right \} } \text{for the FRB rule} 
                \end{matrix} \right. .\numberthis \label{def : allocation}
            \end{align*} 
\end{enumerate}
The focus is to design efficient learning algorithms for the decision-maker, i.e., one specific participant $i$; we can therefore aggregate bids from other participants as the bid of a single \textit{adversary} $\betavec := \bvec_{-i}$ and omit the index $i$ of the learner denoting $\bvec :=\bvec_{i}$.
This setup gives rise to the quasi-linear utility $u(\bvec,\betavec) = \sum_{l=1}^{x(\bvec,\betavec)} \left [ v_l - p(\bvec,\betavec)\right]$. 

In the remainder of this paper, we adopt the LAB pricing rule. The techniques and theoretical proofs can be adapted from one setup to the other with little change. In the absence of specific mention of a pricing rule, our results can be applied to both auction types.

\paragraph{Repeated setting} As mentioned above, this auction is not played just once, but repeated many times (say, each day). We shall then denote a time horizon $T$, and assume a different auction is run at each time step $t \in [T]$ and the objective of the bidder is to maximize their cumulative utility. Quite naturally, the bidder should adjust their bids to the adversary's behavior, learned from the outcomes of the previous iterations. On the other hand, we assume that the bidder does not need to learn their own values, i.e.,  the valuations $(v_k)_{k\in[K]}$ are known to the bidder and do not change over time.

We denote by $\smash{(\bvec^t)_{t \in [T]}}$ and $\smash{(\betavec^t)_{t \in [T]}}$ respectively the sequences of bids of the player and of the adversary, and by $p^t:=p \left ( \bvec^t,\betavec^t \right )$ and $x^t:=x \left ( \bvec^t,\betavec^t \right )$ the price and allocation at time $t$. The utility of the bidder  after the auction $t \in [T]$ is then defined as $u(\bvec^t,\betavec^t) = \sum_{l=1}^{x^t} \left ( v_l - p^t\right)$. As standard in online learning, we evaluate the performance of a learning (bidding) strategy through its \emph{regret}, defined as follows
\begin{equation}
    R_T= \underset{\bvec \in B}{\sup} \sum_{t=1}^T u(\bvec,\betavec^t)- \Expectation[\sum_{t=1}^T u(\bvec^t,\betavec^t)] \; ,\end{equation}
where the expectation is taken over the randomness of the algorithm generating the bids $\bvec^t$. Maximizing the utility of the bidder is equivalent to minimizing the regret.

\paragraph{Feedback}
The bidders can improve their strategy using the information they receive after each iteration of the auction. The type of \emph{feedback} they receive represents their knowledge about the bids of the adversary. In the literature, two common types of feedback are considered \citep{cesa2023role}, 
\begin{enumerate}
    \item the \emph{bandit} feedback where the bidder's allocation $x_t$ is revealed and the price $p_t$ is only revealed if $x_t >0$, and
\item the \emph{full information} feedback where all the bids emitted by all participants are revealed.
\end{enumerate}

Inspired by the terms of commodity electricity markets in several European countries \textit{including Germany and France}, similarly to \cite{karaca2020no}, we shall introduce and study a third partial feedback specific to multi-unit auctions, 
\begin{itemize}
    \item[3.] the \emph{all-winner feedback}: the allocation, the price, and all the winning bids are revealed to the bidder.  
\end{itemize}

\begin{remark} With a uniform discretization of the bidding space $B$, learning to bid in multi-unit uniform auctions can be recast as a special instance of a combinatorial bandit problem. In the latter, the decision maker sequentially selects multiple arms out of $N$ available, i.e., picking at each stage an action in some admissible subset of  $\{0,1\}^N$. Using off-the-shelf combinatorial bandit algorithms, that do not leverage the relevant structure of repeated auctions, would end up in a highly inefficient and sub-optimal procedure (see Section \ref{SE:action}). Our approach is different; in essence, we reduce the complex combinatorial problem by expressing the utility objective as a polynomial (in $K$ and $T$) sum of simpler functions.
\end{remark}

\paragraph{Related Work}

Multiple-unit auctions of indivisible identical items have been extensively studied in their static settings. In particular, how the pricing rules (discriminatory, uniform, VCG) influence revenue \citep{demand_red_ausubel}, social welfare \citep{de2013inefficiency, birmpas2019tight}, or price stability \citep{anderson2018price}. Their use in the context of electricity markets is common and similar questions are being studied with this specific application in mind \citep{cramton2006uniform, fabra2006designing, son2004short, akbari2020electricity}

The repeated setting of auctions, and specifically the use of online learning procedure inspired by Multi Armed Bandits has received lots of attention in the last decade. First studied from the point of view of the auctioneer: \cite{blum2004online} studied maximizing auction revenue, \cite{cesa2014regret} and \cite{kanoria2014dynamic} specifically focused on learning reserve prices. Learning to bid, the bidder's problem, was considered later on, initially in single-item 
 auctions. Second price auctions facing either adversarial or stochastic highest opposing bids were studied in \cite{weed2016online}, and in a contextual, budget-constrained setting by \cite{Flajolet2017RealTimeBW}. \cite{balseiro2019contextual} considered the first price auction with adversarial opposing bids leading to optimal regret rates of $\mathcal{\tilde{O}}(T^{2/3})$ in the known valuation and contextual setting.

First mentioned in \cite{Feng_Learning_to_bid} for unit demand, multiple unit auctions as online learning problems only recently started to be considered as their own topic of interest. Discriminatory pricing and uniform pricing respectively studied in \cite{galgana2023learning} and \cite{branzei_learning_2024} can be learned with $\mathcal{\tilde{O}}(\sqrt{T})$ in the full-information setting. Under bandit feedback, the former achieves $\mathcal{\tilde{O}}(KT^{2/3})$ regret rates in discriminatory pricing and the latter $\mathcal{\tilde{O}}(K^{7/4}T^{3/4})$ regret rates with uniform pricing. Compared to single unit auctions, the combinatorial nature of the action space in $K$-unit auctions makes it a harder learning problem. \cite{branzei_learning_2024} makes use of a cautiously designed equivalent action space represented as a Directed Acyclic Graph (DAG) to address the combinatorial limitation and to design an algorithm guaranteeing the aforementioned regret bounds.

The effects of specific feedback on the ability to achieve lower regret rates have also raised some interest. \cite{Feng_Learning_to_bid} studied the effects of "Win Only" feedback in a more general auction setting. More recently, \cite{cesa2023role} focused on feedback transparency. They characterize gaps in the regret rates that can be achieved depending on the amount of feedback received, getting three separate rates $O(\sqrt{T})$, $O(T^{2/3})$ and $\Omega(T)$ depending on the feedback considered. The work of \cite{karaca2020no}, similarly to the all-winner feedback, studied partial feedback, which lie in between bandit and full-information, motivated by electricity market auctions.

\paragraph{Contribution}
We introduce a novel representation of the action space that overcomes the combinatorial complexity introduced by the multiplicity of the bids in $K$-unit auction. Inspired by the properties of the equivalent action space used by \cite{branzei_learning_2024}, we introduce bid-gaps, to further decompose the utility into a sum of independent functions. This decomposition leads to improved regret rates of $\mathcal{\tilde{O}}(K^{4/3}T^{2/3})$ under bandit feedback, compared to the known upper bound of $\mathcal{\tilde{O}}(K^{7/4}T^{3/4})$. These improved bounds match, in terms of $T$, the rates $\mathcal{\tilde{O}}(KT^{2/3})$ achievable in discriminatory pricing. We notice a reduction to simpler auctions which bear an $\Omega(T^{2/3})$ lower bound on the regret, answering the open question of the optimal rates dependency in $T$ in the bandit setting. Motivated by the terms of bid revelation in electricity reserve markets in several European countries \textit{including Germany and France}, a novel feedback structure is considered, which lies in between bandit feedback and full information. This feedback, which we call all-winner, reveals all the winning bids of the action. We propose an algorithm that achieves a $\mathcal{\tilde{O}}(K^{5/2} \sqrt{T})$ regret, almost matching the regret rates under full information up to a factor $K$, while the lower bound of $\Omega(K\sqrt{T})$ proved by \cite{branzei_learning_2024} for the full-information feedback, naturally extend to this setting. We summarize our results in Table \ref{table : summary} below.
\begin{table}[h!]
  \centering
  \begin{tabular}{llll}
    \toprule
    Feedback      & Literature     & This work & Lower bound\\
    \midrule
    Full information & $\mathcal{\tilde{O}}(K^{3/2}\sqrt{T})$  & $\mathcal{\tilde{O}}(K^{3/2}\sqrt{T})$ &   $\Omega(K\sqrt{T})$\\
    All winner    &    &  $\mathcal{\tilde{O}}(K^{5/2}\sqrt{T})$ & $\Omega(K\sqrt{T})$ \\
    Bandit    & $\mathcal{\tilde{O}}(K^{7/4}T^{3/4})$   & $\mathcal{\tilde{O}}(K^{4/3}T^{2/3})$ & $\Omega(T^{2/3})$$^\star$\\
    \bottomrule
  \end{tabular}\vspace{0.3cm}
  \caption{Regret Rates in multi-unit uniform price auction. $^\star$ holds in the LAB pricing rule setting}
    \label{table : summary}
\end{table}

\section{Action space} \label{SE:action}

We first motivate the new cautiously designed action space, and provide intuitions on how it is constructed and its main properties. We then formalize its definition.

\paragraph{Motivation for an alternative representation}
Usual techniques such as uniform discretization of the action space as in (\cite{Feng_Learning_to_bid}) might lead to consider the subset of non-increasing sequences on this discretization, denoted by $B_\epsilon \subset \smash{\left \{ 0, \epsilon, 2\epsilon, ...,( \floor{\frac{1}{\epsilon}} - 1 ) \epsilon,\floor{\frac{1}{\epsilon}} \epsilon\right \}^K}$. Without loss of generality, we shall assume in the following that $1/\epsilon$ is an integer. The main downside of this representation is that the size of $B_\epsilon$ is exponential and thus, without any further properties of the problem leveraged, this would lead to arbitrarily bad regret rates $\tilde{\mathcal{O}}(T^\frac{K+1}{K+2})$ in bandit setting. Even though we can restrict the available action to \emph{reasonable ones} (ie undominated strategies \ref{lemma : dominated strategies}) this isn't enough to achieve improved rates in general.

\cite{branzei_learning_2024} proposed a Directed Acyclic Graph (DAG) equivalent of the action space $B_\epsilon$  to overcome this combinatorial limitation. They use the decomposition of the utility into a sum of independent functions, depending only on pairs of consecutive bids (the edges in their graphs), to reduce the combinatorial complexity to only 2 (instead of $K$), and thus achieved an $\tilde{O}(K^{7/4}T^{3/4})$ regret bound under bandit feedback. 
Motivated by the breakthrough enabled by such a representation of the action space, we consider a new equivalent action space $H_\epsilon$, introduced in \autoref{eq: def_h_epsilon}. This new action space allows to leverage more precisely the regularity of the utility with respect to the bidder's choice of bids, which in turn leads to improved regret bounds under bandit feedback, presented in Theorem~\ref{theorem : Regret exp3}.

\subsection{Action space tailored to the outcomes }

We now provide intuitions on the utility regularity that will be leveraged. We start by observing that, for a given auction, the price is either set by one of the bidder's bid or by a bid from the adversary. 

Assume that the bidder bids $(b_1,...,b_K) \in B_\epsilon$, and that the price is $b_k$ for some $k\in[K]$. Then, we claim that many bid profiles would have led to the same outcome.
Indeed, any bid profile (from the bidder) with the same $k^\text{th}$ bid $b_k$, 
leads to the same outcome. In the alternative case, where the adversary sets the price, if the bidder wins $k$ items (i.e., the $K-k$ adversary's bid $\beta_{K-k}^t$ sets the price), then any bid profile satisfying \textit{$b_k \geq \beta_{K-k}^t \geq b_{k+1}$}  leads to the same outcome.

Notice that in the two aforementioned cases of regularities of the utility, all bids $\bvec \in B_\epsilon$ which would lead to the same outcome share one of the following properties: for a specific $k$ and $j$,\begin{itemize}[parsep=0cm,itemsep=0.0cm,topsep=0cm]
\item in the first case :  $ b_k = j \epsilon $,
\item in the second case : $ b_k \geq (j+1) \epsilon \geq \beta_{K-k}^t \geq j \epsilon \geq b_{k+1} $. \end{itemize}

For simplicity, we shall assume that the bids of the decision-maker belong to the $\epsilon$-discretization, i.e.,  $(b_1,...,b_K) \in B_\epsilon$ while the bids of the adversary do not belong to it, in order to avoid ties\footnote{This assumption comes without loss of generality, since e.g. adding uniform noise with extremely small variance to the bidder's bid prevent ties a.s., see \autoref{lemma : avoiding ties regret} in \autoref{appendix : avoiding ties}}. We denote this set $B_{\setminus \epsilon}$, the set of non-increasing sequences of $[0,1]^K$ without values in $\left [\frac{1}{\epsilon}\right ]$.

We, introduce an alternative description of the bidding space $B_\epsilon$ whose structure closely matches the aforementioned regularities' in order to improve the bidder's strategy. It leverages new binary variables indicating which of the aforementioned properties a bid $\bvec \in B_\epsilon$ has, they are defined as follows: for any $k \in [K]$ and $j \in \left [  \frac{1}{\epsilon}  \right  ]$,  \begin{align}
    h_{k,j}(\bvec) &= \indicator[ b_k=j \epsilon ] \label{eq: binary_variables _ interger}\\
    h_{k+\frac{1}{2},j} (\bvec) & = \indicator[ b_k \geq (j+1) \epsilon >  j \epsilon \geq b_{k+1} ].
\end{align}

Let $\mathcal{K}=\left \{1,\frac{3}{2},2,...,K-1,\frac{2K-1}{2K},K \right \}$ and $\mathcal{J}_\epsilon=\left [  \frac{1}{\epsilon}  \right  ]$. For any $\bvec \in B_\epsilon$, we define the pseudo-bid $\hvec_\bvec$ to be the list of these binary variable $h_{k,j}$ with $k,j \in \mathcal{K} \times \mathcal{J}$, such that $h_{k,j}(\bvec) = 1$, ordered in lexicographic order, increasingly in $k$ and decreasingly in $j$.  We naturally define $H_\epsilon$ the pseudo-bid space generated by the bid space $B_\epsilon$  :\begin{equation} 
       H_\epsilon =\left \{\hvec_\bvec \middle | \bvec \in B_\epsilon \right \} \label{eq: def_h_epsilon} \end{equation}


\begin{lemma} \label{lemma : bijective map}
    For each pseudo-bid $\hvec \in H_\epsilon$, there exists a unique $\bvec \in B_\epsilon$ such that $\hvec=\hvec_\bvec$. This therefore defines a bijective mapping between $H_\epsilon$ and $B_\epsilon$. 
\end{lemma}
\begin{proof}[Proof]
From the expression of $H_\epsilon$ in \eqref{eq: def_h_epsilon}, it is clear that the mapping $\bvec \mapsto \hvec_\bvec$ , is surjective. \newline
Let $\hvec \in H_\epsilon$, there exists $\bvec=\{b_1,...,b_K\} \in B_\epsilon$ such that $\hvec=\hvec_\bvec$. Let $j_k \in \mathcal{J}_\epsilon$ such that $b_k=j_k \epsilon$, for all $k \in [K]$, we have $h_{k,j_k} \in \hvec$.
    If there exists another bid $\tilde{\bvec}=\{\tilde{b}_1,...,\tilde{b}_K\}  \in B_\epsilon$ such that $\hvec=\hvec_{\tilde{\bvec}}$, for all $k \in [K], h_{k,j_k} (\tilde{\bvec})=1$.Therefore \eqref{eq: binary_variables _ interger} yields $\tilde{b}_k=j_k \epsilon =b_k$ for all $k \in [K]$. This proves unicity and therefore that the mapping is bijective.
\end{proof}



The following characterization of the pseudo-bid space directly follows from \autoref{lemma : bijective map}.

\begin{corollary} \label{corollary : action space}
Given a bid $\mathbf{b}=(b_i)_{i \in [K]} \in B_\epsilon$ and the pseudo-bid  $\mathbf{h}_\bvec \in H_\epsilon$, we have the following: for all $k,j \in [K] \times \mathcal{J}_\epsilon $, 
 \begin{align}
     &b_k=j\epsilon \iff h_{k,j} \in \mathbf{h}_\bvec \label{eq : corollary bids}\\
         & b_k \geq (j+1)\epsilon > j \epsilon \geq b_{k+1} \iff h_{k+\frac{1}{2},j} \in \mathbf{h}_\bvec  \label{eq : corollary pseudo-bids}
 \end{align}
     
\end{corollary}

To provide further intuition,  Figure~\ref{figure :  graph of paths nat} and Figure~\ref{figure : graph of pathsnot nat} show how two corresponding bids might be represented in $B_\epsilon$ and $H_\epsilon$.  The bids \eqref{eq : corollary bids} are represented by circles, while the bid-gaps \eqref{eq : corollary pseudo-bids} are ellipses.

\begin{figure}[h!] 
 \begin{subfigure}[b]{0,4\textwidth} 
        \tikzset{every picture/.style={line width=0.75pt}} 

\begin{tikzpicture}[x=0.75pt,y=0.85pt,yscale=-0.55,xscale=0.55]

\draw    (199,470.4) -- (199,721.4) (205,520.4) -- (193,520.4)(205,570.4) -- (193,570.4)(205,620.4) -- (193,620.4)(205,670.4) -- (193,670.4)(205,720.4) -- (193,720.4) ;
\draw [shift={(199,721.4)}, rotate = 270] [color={rgb, 255:red, 0; green, 0; blue, 0 }  ][line width=0.75]    (0,5.59) -- (0,-5.59)   ;
\draw [shift={(199,470.4)}, rotate = 270] [color={rgb, 255:red, 0; green, 0; blue, 0 }  ][line width=0.75]    (0,5.59) -- (0,-5.59)   ;
\draw    (257,471) -- (257,722) (263,521) -- (251,521)(263,571) -- (251,571)(263,621) -- (251,621)(263,671) -- (251,671)(263,721) -- (251,721) ;
\draw [shift={(257,722)}, rotate = 270] [color={rgb, 255:red, 0; green, 0; blue, 0 }  ][line width=0.75]    (0,5.59) -- (0,-5.59)   ;
\draw [shift={(257,471)}, rotate = 270] [color={rgb, 255:red, 0; green, 0; blue, 0 }  ][line width=0.75]    (0,5.59) -- (0,-5.59)   ;
\draw    (318,471) -- (318,722) (324,521) -- (312,521)(324,571) -- (312,571)(324,621) -- (312,621)(324,671) -- (312,671)(324,721) -- (312,721) ;
\draw [shift={(318,722)}, rotate = 270] [color={rgb, 255:red, 0; green, 0; blue, 0 }  ][line width=0.75]    (0,5.59) -- (0,-5.59)   ;
\draw [shift={(318,471)}, rotate = 270] [color={rgb, 255:red, 0; green, 0; blue, 0 }  ][line width=0.75]    (0,5.59) -- (0,-5.59)   ;
\draw    (378,471) -- (378,722) (384,521) -- (372,521)(384,571) -- (372,571)(384,621) -- (372,621)(384,671) -- (372,671)(384,721) -- (372,721) ;
\draw [shift={(378,722)}, rotate = 270] [color={rgb, 255:red, 0; green, 0; blue, 0 }  ][line width=0.75]    (0,5.59) -- (0,-5.59)   ;
\draw [shift={(378,471)}, rotate = 270] [color={rgb, 255:red, 0; green, 0; blue, 0 }  ][line width=0.75]    (0,5.59) -- (0,-5.59)   ;
\draw  [color={rgb, 255:red, 255; green, 255; blue, 255 }  ,draw opacity=0 ][fill={rgb, 255:red, 74; green, 144; blue, 226 }  ,fill opacity=0.53 ][line width=0.75]  (187.52,520.92) .. controls (187.52,514.23) and (192.86,508.8) .. (199.44,508.8) .. controls (206.02,508.8) and (211.35,514.23) .. (211.35,520.92) .. controls (211.35,527.61) and (206.02,533.04) .. (199.44,533.04) .. controls (192.86,533.04) and (187.52,527.61) .. (187.52,520.92) -- cycle ;
\draw  [color={rgb, 255:red, 255; green, 255; blue, 255 }  ,draw opacity=0 ][fill={rgb, 255:red, 74; green, 144; blue, 226 }  ,fill opacity=0.53 ][line width=0.75]  (245.52,620.92) .. controls (245.52,614.23) and (250.86,608.8) .. (257.44,608.8) .. controls (264.02,608.8) and (269.35,614.23) .. (269.35,620.92) .. controls (269.35,627.61) and (264.02,633.04) .. (257.44,633.04) .. controls (250.86,633.04) and (245.52,627.61) .. (245.52,620.92) -- cycle ;
\draw  [color={rgb, 255:red, 255; green, 255; blue, 255 }  ,draw opacity=0 ][fill={rgb, 255:red, 74; green, 144; blue, 226 }  ,fill opacity=0.53 ][line width=0.75]  (306.52,620.92) .. controls (306.52,614.23) and (311.86,608.8) .. (318.44,608.8) .. controls (325.02,608.8) and (330.35,614.23) .. (330.35,620.92) .. controls (330.35,627.61) and (325.02,633.04) .. (318.44,633.04) .. controls (311.86,633.04) and (306.52,627.61) .. (306.52,620.92) -- cycle ;
\draw  [color={rgb, 255:red, 255; green, 255; blue, 255 }  ,draw opacity=0 ][fill={rgb, 255:red, 74; green, 144; blue, 226 }  ,fill opacity=0.53 ][line width=0.75]  (366.52,670.92) .. controls (366.52,664.23) and (371.86,658.8) .. (378.44,658.8) .. controls (385.02,658.8) and (390.35,664.23) .. (390.35,670.92) .. controls (390.35,677.61) and (385.02,683.04) .. (378.44,683.04) .. controls (371.86,683.04) and (366.52,677.61) .. (366.52,670.92) -- cycle ;
\draw [color={rgb, 255:red, 74; green, 144; blue, 226 }  ,draw opacity=0.5 ][fill={rgb, 255:red, 255; green, 255; blue, 255 }  ,fill opacity=1 ][line width=3.75]    (205.4,529.9) -- (248.01,606.78) ;
\draw [shift={(251.4,612.9)}, rotate = 241] [fill={rgb, 255:red, 74; green, 144; blue, 226 }  ,fill opacity=0.5 ][line width=0.08]  [draw opacity=0] (14.38,-6.91) -- (0,0) -- (14.38,6.91) -- cycle    ;
\draw [color={rgb, 255:red, 74; green, 144; blue, 226 }  ,draw opacity=0.5 ][fill={rgb, 255:red, 255; green, 255; blue, 255 }  ,fill opacity=1 ][line width=3.75]    (269.35,620.92) -- (299.52,620.92) ;
\draw [shift={(306.52,620.92)}, rotate = 180] [fill={rgb, 255:red, 74; green, 144; blue, 226 }  ,fill opacity=0.5 ][line width=0.08]  [draw opacity=0] (14.38,-6.91) -- (0,0) -- (14.38,6.91) -- cycle    ;
\draw [color={rgb, 255:red, 74; green, 144; blue, 226 }  ,draw opacity=0.5 ][fill={rgb, 255:red, 255; green, 255; blue, 255 }  ,fill opacity=1 ][line width=3.75]    (327.4,627.9) -- (364.15,660.27) ;
\draw [shift={(369.4,664.9)}, rotate = 221.38] [fill={rgb, 255:red, 74; green, 144; blue, 226 }  ,fill opacity=0.5 ][line width=0.08]  [draw opacity=0] (14.38,-6.91) -- (0,0) -- (14.38,6.91) -- cycle    ;

\draw (393.82,451.17) node [anchor=north west][inner sep=0.75pt]  [font=\small,rotate=-359.2]  {$\left\lfloor \frac{1}{\epsilon }\right\rfloor \epsilon $};
\draw (402.82,712.17) node [anchor=north west][inner sep=0.75pt]  [font=\small,rotate=-359.2]  {$\epsilon $};
\draw (395.82,663.17) node [anchor=north west][inner sep=0.75pt]  [font=\small,rotate=-359.2]  {$2\epsilon $};
\draw (191,739.8) node [anchor=north west][inner sep=0.75pt]    {$b_{1}$};
\draw (249,740.8) node [anchor=north west][inner sep=0.75pt]    {$b_{2}$};
\draw (371,739.8) node [anchor=north west][inner sep=0.75pt]    {$b_{K}$};

\end{tikzpicture}
        \caption{Bids in $B_\epsilon$}
        \label{figure : graph of paths nat}
    \end{subfigure}
 \begin{subfigure}[b]{0,5\textwidth} 
        \input{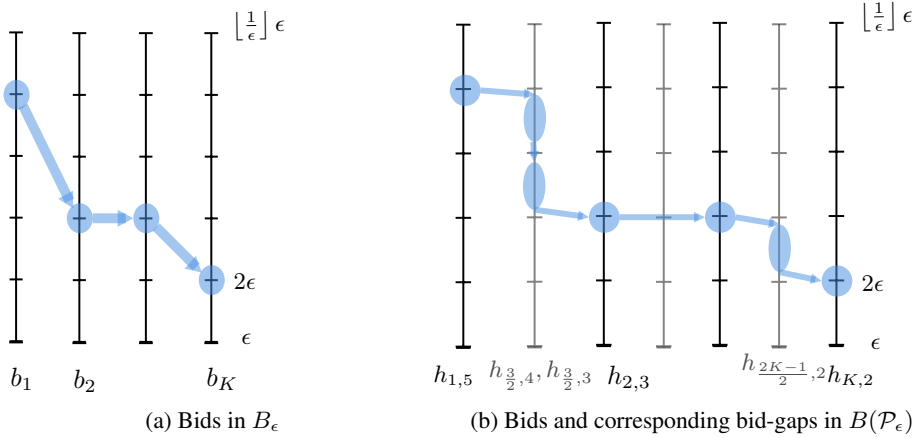}
        \caption{Bids and corresponding bid-gaps in $B(\Prix_\epsilon)$}
        \label{figure : graph of pathsnot nat}
    \end{subfigure}
\caption{Graph representation of action spaces $B_\epsilon$ (\cite{branzei_learning_2024}) and $B(\Prix_\epsilon)$ (this paper)}
\label{figure:   both actions}
\end{figure}

\subsection{Utility decomposition}

Leveraging the new action space, we define the utility, price, and allocation function on $H_\epsilon$ resulting from the bijective map with $B_\epsilon$.
Let $\hvec \in H_\epsilon$ and $\bvec \in B_\epsilon$ the unique element of $B_\epsilon$ such that $\hvec = \hvec_\bvec$. For all $\betavec \in B_{\setminus \epsilon}$, we define the utility as $u_H(\hvec_\bvec,\betavec) :=  u(\bvec,\betavec)$, the price $x_H(\hvec_\bvec,\betavec) :=  x(\bvec,\betavec)$ and the allocation $p_H(\hvec_\bvec,\betavec) :=  p(\bvec,\betavec)$

The following additional set notation, which matches a binary variable $h_{k,j}$ to its corresponding price range, allows for unified descriptions : 
\begin{equation}\Prix_\epsilon(h_{k,j}) : = \begin{matrix*}
    \{ j \epsilon \} & \text{if $k$ is an integer} \\
     \left  (j \epsilon, (j+1) \epsilon \right ) & \text{if $k$ is half integer\ ,}
\end{matrix*}\end{equation}

We now explicitly show how the pseudo-bid space is \emph{well suited} to capture the regularity of the outcomes of the auction (and therefore of the utility) mentioned above. To be more precise, the following \autoref{lemma : regularity over outcomes} states that within a pseudo-bid profile $\hvec$, a pseudo-bid $h_{k,j}$ can be \emph{credited} for the outcome: any other pseudo-bid profile containing this pseudo bid would have lead to the same outcome.

\begin{lemma} \label{lemma : regularity over outcomes}
    Let $\betavec \in B_{\setminus \epsilon}$ and $(k,j) \in \mathcal{K}\times \mathcal{J}_\epsilon$. There exists $C \in \{0,1\}$, such that for all $\hvec \in H_\epsilon$ with $h_{k,j} \in \hvec$, \begin{equation}
        \indicator \{ p_H(\hvec,\betavec) \in \Prix_\epsilon(h_{k,j}) \} \cap \{ x_H(\hvec,\betavec) = \floor{k} \} = C 
    \end{equation}
    and if $C=1$, $p_H(\hvec,\betavec)$ is also constant on $\left \{ \hvec \in H_\epsilon : h_{k,j} \in \hvec \right \}$
\end{lemma}

\begin{proof}
     Let $\bvec \in B_\epsilon$ such that $h_{k,j} \in \hvec_\bvec$. If $k$ is integer, then \autoref{corollary : action space} yields $b_k=j \epsilon$, hence we get $\indicator \{ p_H(\hvec,\betavec) = j\epsilon \} \cap \{ x_H(\hvec,\betavec) = k \} = \indicator \{ \beta_{K-k} >  j\epsilon = b_k > \beta_{K-k+1} \} $ which only depends on $k,j$ and $\betavec$. If $k$ is half-integer, then \autoref{corollary : action space} yields $b_{k+1} < j \epsilon < (j+1) \epsilon < b_{k}$, hence we get $\indicator \{ p_H(\hvec,\betavec) = \Prix_\epsilon(h_{k,j}) \} \cap \{ x_H(\hvec,\betavec) = k \} = \indicator \{ j\epsilon < \beta_{K-k} <  (j+1) \epsilon  \} $ which also only depends on $k,j$ and $\betavec$.
     It is straightforward to see that when the indicator function takes value one, the price is constant with value $j \epsilon$ for the integer case and $\beta_{K-k}$ in the half integer case.
\end{proof}
 Lemma~\ref{lemma : regularity over outcomes} allows us to exhibit a key property of the utility on the pseudo-bid space: it can be decomposed into a sum of sub-utilities (defined in \autoref{def : subutilities}), each of which only depends on one of the components of $\hvec$.

\begin{lemma}\label{lemma : decomposition}
    Let $\mathbf{h} \in H_\epsilon$ and $\betavec \in B_{\setminus \epsilon}$. The utility of the bidder rewrites as a sum of sub-utilities:
 \begin{equation}
u_H(\hvec,\betavec)=\sum_{h_{k,j} \in \mathbf{h}} w(h_{k,j},\betavec), \; \text{with} \label{eq : decompositon of utility a.s.}\end{equation}
\begin{equation} w(h_{k,j},\betavec)  := \indicator\left \{ \left \{ p_H(\mathbf{h},\betavec)  \in \Prix_\epsilon(h_{k,j}) \right \} \cap \left \{ x_H(\mathbf{h},\betavec)=\floor{k} \right  \} \right \} \sum\limits_{l=1}^{\floor{k}} (v_l-p_H(\hvec,\betavec)) . \label{def : subutilities}  \end{equation}
\end{lemma}
\begin{proof}[Proof of Lemma~\ref{lemma : decomposition}]
$x(\hvec, \betavec) = 0$ implies that $u$ as defined in Equation~(\ref{eq : decompositon of utility a.s.}) is zero, as expected. For the case where $x(\hvec, \betavec) > 0$, we use that the indicator functions in Equation~\ref{def : subutilities} correspond to disjoint events. Thus, there exists a unique pair $(k,j)\in \mathcal{K}\times \mathcal{J}_\epsilon$ such that $w(h_{k,j}, \betavec)>0$, furthermore, this sub-utility $w(h_{k,j}, \betavec)$ matches the utility $u(\hvec, \betavec)$. This concludes the proof. 
\end{proof}

While the expression of these indicators in Equation~\eqref{def : subutilities} involves the full action $\mathbf{h}$,
Lemma \ref{lemma : regularity over outcomes} shows that they only depend either on the associated bid $h_{k,j} $.
Furthermore, notice that within a given $\hvec \in H_\epsilon$, the events corresponding to each $h_{k,j} \in \hvec$ are disjoints and therefore only one can be realized. As a result there is at most one $h \in \hvec$ with positive sub-utility, we denote it $h_\star(\hvec,\betavec)$.


\section{Learning algorithms and guarantees}

We first present two algorithms and estimators of the utility corresponding to the different feedback settings. We then state the regret rate achieved by these algorithms when used with the introduced estimator, depending on the setting considered, as well as a lower bound on the regret in the bandit setting. To simplify notations, we shall further introduce $u_H^t(.) :=u_H(.,\betavec^t)$ and $w^t(.) :=w(.,\betavec^t)$. 

\subsection{Algorithm for online learning in K-unit uniform auction}

The regret minimizing algorithm combines two separate procedures,
\begin{itemize}
    \item An exponential weight update, detailed in Algorithm~\ref{algorithm : exp3 K bis}, which creates and updates the weights for each bid and bid-gap at each time step, akin to an EXP3 type algorithm (but non-normalized) (\cite{cesa2006prediction},\cite{lattimore2020bandit}). 
    \item A sampling procedure, tailored to our action space, described in Algorithm~\ref{alg:weight_pushing},  that re-normalizes the weights into a probability distribution by using a weight-pushing method, and uses an efficient procedure to sample an action. (\cite{goos_path_2002}). 
\end{itemize} 

Because of the combinatorial structure of the problem, using directly the  exponential weight algorithm would be highly inefficient, with a complexity of order $\mathcal{O}(\frac{1}{\epsilon^K})$ to store weights and compute probabilities directly on the action space. The second auxiliary sampling algorithm gets rid off this complexity burden. 

Algorithm~\ref{algorithm : exp3 K bis}'s pseudo-code details the exponential weight algorithm used as a building block of the no-regret procedure. It computes weights for each pseudo-bid based on the corresponding sub-utilities, or their estimated values. These weights are used to sample a pseudo-bid sequence $\hvec \in H_\epsilon$.

\begin{algorithm}
    \caption{Component based exponential weighting}
    \label{algorithm : exp3 K bis}
    \textbf{Input:} time horizon $T$,  parameters $\varepsilon>0$ and $\eta>0$
    
    \textbf{Output:} actions for each time step $\left ( \hvec^1,\hvec^2,\hdots, \hvec^{T-1},\hvec^T \right ) \in (H_\epsilon)^T $.

    \textbf{Initialize:} for $(k,j) \in \mathcal{K} \times \mathcal{J}_\epsilon$, set $ \HalfProba^0 (h_{k,j}) =1 $

    \For{$t = 1,2,\hdots, T$}{
     Sample $\hvec^t$ using the sampling procedure of Algorithm~\ref{alg:weight_pushing}, with input parameters $\HalfProba^{t-1}$\;
    Receive the utility $u^t = u(\hvec^t, \betavec^t)$ and the feedback. Based on this feedback, define
        $v^t=\left\{\begin{array}{ll} w^t = w(., \betavec^t) & \text{in the full-information feedback, see  }  \eqref{def : subutilities}\\  
         \hat{w}^t & \text{with bandit feedback, see } \eqref{def : subutility bandit estimator}\\ \bar{w}^t& \text{with  all-winner  feedback, see } \eqref{def : subutilities all winner estimator }\end{array}\right.$\;
     Update the weights based on the weights $v^t$ :
     
     \For{ $(k,j) \in \mathcal{K} \times \mathcal{J}_\epsilon$}{
            \begin{align} \label{eq : proba update rule}
                    & \HalfProbat (h_{k,j})= \HalfProba^{t-1} (h_{k,j}) \exp\left ( \eta v^t(h_{k,j}) \right ) 
            \end{align}}
}
\end{algorithm}

Algorithm~\ref{alg:weight_pushing} is a sampling procedure and uses a weight-pushing technique (\cite{goos_path_2002}) to efficiently sample pseudo-bids and compute weights and probabilities on $H_\epsilon$. It is exploiting the lexicographical ordering of the pseudo bid (increasingly in $k$ and decreasingly in $j$) which creates a graph-like structure, as illustrated in Figure~\ref{figure : graph of pathsnot nat}. It is indeed possible to sample an element of $H_\epsilon$ by repeatedly sampling the next binary variable conditionally on the previous one. 
To explicit this graph-like structure, we define $s(.)$ the successor function, which given a pseudo-bid $h_{k,j}$ provides the set of possible next element of an action.
For $k \in [K-1]$, $j \in \mathcal{J}_\epsilon \setminus \{ 0 \}$, 
\begin{equation}
    s(h_{k+1/2,j}) := \left \{ 
        h_{k+1/2,j-1}, h_{k+1,j} \right \}
\end{equation}
\begin{equation}
    s(h_{k,j}) := \left \{ 
        h_{k+1/2,j-1}, h_{k+1,j} \right \}
\end{equation}
\begin{equation}
    s(h_{k,0}) := \left \{ h_{k+1,0} \right \}
\end{equation}
As well as the following, which serves as the stopping condition of the sampling:
\begin{equation}
    \forall j \in \mathcal{J}_\epsilon,\  s(h_{K,j}) := \emptyset
\end{equation}

\begin{algorithm}[h!] 
		\caption{Selection of the bids by a weight-pushing algorithm}
            \label{alg:weight_pushing}
		\textbf{Input:} Weights $\HalfProba$ for every bids or bid-gap $(h_{k,j})_{(k,j)\in \mathcal{K} \times \mathcal{J}_\epsilon }$\;
		\textbf{Output:} bid vector $\hvec \in H_\epsilon $\;
            \textbf{Initialize :}$k\gets K-\frac{1}{2}$. For all $j \in \mathcal{J}_\epsilon $, $\Gamma(h_{K,j})\gets 1$\;
            \While{$k \geq 1$}{  
			\For{$j \in \mathcal{J}_\epsilon$}{					  
$ \Gamma(h_{k,j}) \gets \sum_{h \in s(h_{k,j})} \Gamma(h) \HalfProba(h)
                        $}
   $k \gets k-\frac{1}{2}$}

      $\Gamma_0\gets \sum_{j \in \left [ \frac{1}{\epsilon} \right ]} \HalfProba(h_{1,j}) \Gamma(h_{1,j})
                    \label{def: gamma 0}$\;
        Sample $h$ according to the probabilities:\\ $\forall j \in \mathcal{J}_\epsilon, \Proba(h=h_{1,j}) = \HalfProba(h_{1,j})\frac{\Gamma(h_{1,j})}{\Gamma_0}$\;
      $\hvec \gets \{h\}$\;
         \While{$s(h) \neq \emptyset$}{sample $h_+ \in s(h)$, the next element of the sequence $\hvec$,  with probability: 
                    $
                        \Proba(h_+ =h'| h) =\HalfProba(h') \frac{\Gamma(h')}{\Gamma(h)}$ for all $h' \in s(h)$\;
             $\hvec \gets \hvec \cup \{h_+\}$\;           
            $h \gets h_+$\;
            }
        \end{algorithm}

The combination of Algorithm~\ref{algorithm : exp3 K bis} and Algorithm~\ref{alg:weight_pushing} leads to the following probabilities, typical of an exponential weight algorithm, on $H_\epsilon$. For $\hvec \in H_\epsilon$,
\begin{align} 
    \Proba^t(\hvec) & =  \frac{ \exp\left ( \sum_{n=0}^{t} {\eta u_H^n(\hvec)} \right ) }  {\sum_{\mathbf{a} \in H_\epsilon} \exp \left (\sum_{n=0}^{t}  \eta u_H^n(\mathbf{a})\right )} \label{eq : probability of exponential weight update generic recur} 
\end{align}
as shown in \autoref{app :learning}, in \autoref{eq : probability of exponential weight update}.

\paragraph{Estimators}
With partial feedback (either bandits or all-winner), the bidder does not gather enough information to compute all of the sub-utilities, and it can only do it for a subset of pseudo-bids.
They must therefore resort, as it is standard in multi-armed bandit literature, to estimators that should leverage all the information available. Under bandit feedback, only sub-utilities of binary variables which belong to the action played at time $t$ can be computed.- On the other hand, under all-winner feedback, the richer feedback allows to compute sub-utilities for a bigger set of binary variables $h$, we denote it $\Observed$ and define it in \autoref{lemma : observed all-winner}.

\begin{restatable}{lemma}{observedAllWinner}
 \label{lemma : observed all-winner}
    With the all-winner feedback, the bidder can compute from its feedback the sub-utilities of any pseudo bid 
 in $\Observed(\hvec^t,\betavec^t)$,  defined as:
\begin{equation}
        \Observed(\hvec^t,\betavec^t) :=
            \left \{ h_{k,j}, (k,j) \in \mathcal{K} \times \mathcal{J}_\epsilon  \middle |  \text{ s.t. } \ \{ k > x^t \} \  \text{or} \ \{  k = x^t \text{and} \ j \epsilon \geq p^t \} \right   \} 
    \end{equation}
    Where $x^t:=x_H(\hvec^t,\betavec^t)$ and $p^t=p_H(\hvec^t,\betavec^t)$. 
\end{restatable}
The formal proof of Lemma~\ref{lemma : observed all-winner} is in Appendix~\ref{proof}. 

    As noted above, within a given pseudo-bid $\hvec$, only one sub-utility can be non-zero. We therefore also define the set of binary variables with non-zero sub-utilities $\Observed_\star(\hvec^t,\betavec^t) := \{ h_{k,j} \in  \Observed(\hvec^t,\betavec^t)| w(h_{k,j},\betavec^t) >0 \}$.

We can now formally introduce the estimators used by the no-regret procedure.

\begin{definition}[Estimators]
Let $\hvec^t \in H_\epsilon$ be the action played by the learner, and $\betavec^t \in B_{\setminus \epsilon}$ the bids of the adversary at time $t$,. For any bid or bid-gap $h$, we define the sub-utility estimators: 
\begin{align*}
 \text{Bandit feedback }& \hat{w}^t(h) = \indicator(h =  h^t_\star) \frac{w^t(h) - K}{\Probat(h)}, \numberthis \label{def : subutility bandit estimator} \\
 \text{All-winner feedback } & \bar{w}^t (h) = \indicator \left ( h \in \Observed^t_\star( \hvec^t) \right ) \frac{w^t(h)-K}{\underset{\mathbf{f}^t \sim \mathcal{B}^t}{\Proba}(h \in \Observed_\star^t (\mathbf{f}^t) )} \numberthis \label{def : subutilities all winner estimator }
\end{align*}
where $h^t_\star := h_\star(\hvec^t,\betavec^t)$ is the sub/pseudo-bid played at time $t$ that has non-zero sub-utility, $\mathcal{B}^t$ is the probability distribution on $H_\epsilon$ as in \eqref{eq : probability of exponential weight update generic recur} and $\Probat(h) := \sum_{\hvec \in H_\epsilon : h \in \hvec} \Probat(\hvec) $
\end{definition}
Naturally, estimation of the utility of any action 
$\hvec \in H_\epsilon$ is done with a simple summation,  over $h \in \hvec$,  of these estimates. 
\subsection{Regret Analysis}

Combining Algorithm \ref{algorithm : exp3 K bis} and the sampling Algorithm \ref{alg:weight_pushing}, we recover the regret guarantees obtained by (\cite{branzei_learning_2024}, Theorem 2) in the same setting for the full-information feedback. We provide a formal statement and proof of this result in the \autoref{app :learning}, in \autoref{theorem : full info regret}. We now analyze the performance of the learning procedure for the two other types of feedback.


 \begin{theorem} \label{theorem : Regret exp3}
   In the repeated $K$-unit auction with uniform pricing guarantees and under bandit feedback, Algorithm \ref{algorithm : exp3 K bis}  incurs a regret of at most $\mathcal{O} \left ( K^{4/3} T^{2/3} \log \left (T \right ) \right ) $ For any time horizon $T$, with the choices of $\epsilon = \left ( \frac{K}{T} \right )^{1/3}$ and  $\eta = K^{-1/3}T^{-2/3}\sqrt{\log \left ( \frac{T}{K} \right )/3}$.
 \end{theorem}
\begin{proof}[Proof sketch]    
The aforementioned regret bounds are proved by using a similar analysis as the one in \cite{lattimore2020bandit} to obtain regret bounds of EXP3 algorithm in the adversarial bandits case. We apply this analysis to the discretized action space $H_\epsilon$ and bound the additional cost of using a discretization separately. Then we choose a discretization size $\epsilon$ to minimize the total regret. 

The improvement over known regret bounds in \cite{branzei_learning_2024} results from the decomposition of the utility into sub-utilities. Since these sub-utilities only depend on one bid or bid-gap, the \textit{variance} of the estimators (cf Lemma \ref{lemma : variance bandits estimator}) only depend on the possible number of bids or bid-gaps ( of order $\frac{1}{\epsilon}$) not the number of bid profiles ( of order $\frac{1}{\epsilon^K}$). This is akin to why combinatorial bandits under semi-bandit feedback (\cite{audibert2014regret}) achieve better regret than under bandit-feedback. 
\end{proof}
  \begin{theorem} \label{theorem : Regret all winner}
   For any time horizon $T$, using Algorithm \ref{algorithm : exp3 K bis} in the repeated $K$-unit auction with uniform pricing guarantees, under all-winner feedback, a regret of at most $\smash{\mathcal{O} \left ( K^{5/2} \sqrt{T} \log (T) \right ) }$ with $\eta =K^{-1}T^{-1/2}$ and  $\epsilon = K^{3/2}T^{1/2}$.
 \end{theorem}

\begin{proof}[Proof sketch]
    The proof of these bounds in the all-winner feedback follows closely the analysis used for the bandit feedback. 
    Better bounds are achieved in comparison to the bandit's feedback thanks to the lower \textit{variance}  of the estimator defined in \eqref{lemma : variance all-winner estimator}, proved in Lemma~\ref{lemma : variance all-winner estimator}. Intuitively, this comes from the ability to observe the realized utility more often, allowed by the richer feedback. 
    We exhibit this by using tools used in bandits with graph feedback \citep{Alon_Graph-Structured_Feedback}.
\end{proof}

\paragraph{Regret lower bound}
We provide a matching lower bound on the regret of any online learning algorithm (Lemma~\ref{lemma : lower bound}), in the bandit feedback setting, by extending a result from \citep{balseiro2019contextual} in the context of single price auctions. 
This partially answers an open question raised by \cite{branzei_learning_2024} regarding the achievable learning rate in the bandit setting for the problem that we consider. 

\begin{lemma} \label{lemma : lower bound}
    Any online learning procedure must incur $\Omega(T^{2/3})$ regret in multi-unit uniform auction with the Last Accepted Bid rule under bandit feedback Bid pricing rule.
\end{lemma}

This stems from the fact that against an adversary that only plays bids with value 1 except for its last bid, the auction is essentially a first-price auction. 

\begin{proof}[Proof] We extend the lower bound on the regret of the first price auction in \cite{balseiro2019contextual}. At time $t$, let $\betavec^t=\{1,1,...,1,h^t\}$ be the bid of the adversary, let the valuation of the learner be $v=(1,0,..,0)$ and denote $\bvec^t = (b_1^t,...,b_k^t) \in B$ the learner's bid.
We only consider sensible bids, such that $b_i^t > 0 \iff i=1$, because they are dominating strategies.
The learner's utility is $u(\hvec,\betavec) = \indicator[b_1^t > h^t] (1-b_1^t)$, and the bandits feedback, ($x^t,\indicator[x^t >0] p^t)$ is : $(\indicator[ b_1^t > h^t],\indicator[ b_1^t > h^t] b_1^t)$ which coincides with both the utility and the bandits feedback of the first price auction with value 1. 

This specific instance of the repeated $K$-unit auction therefore coincides with the repeated first price auction with opposing bid $h^t$ at time $t$, which can be any instance of the first price auction.
Therefore if no learning algorithm can guarantee better regret than $\mathcal{O}(T^{2/3})$ in the latter problem, no algorithm can guarantee better regret than $\mathcal{O}(T^{2/3})$ in the former.

\end{proof}

\section{Conclusion}

We provided the first no-regret algorithm achieving optimal rates in $T$ for the $K$-unit uniform auction under bandit, full information, and all winner feedback. The techniques and theoretical tools presented can be applied to obtain similar regret guarantees in the adversarial bid setting with random valuation, under the assumption that the valuation and opposing bids are independent. An interesting open question is whether similar rates can be achieved in a contextual setting (when valuations changing at each round are observed before each play). 
The obtained regret rates match the ones obtained in the discriminatory price auction, a commonly compared auction mechanism, up to a factor $\mathcal{O}(K^{\frac{1}{3}})$. This raises the question of whether this gap can be closed or if a lower bound showing a separation in achievable regret rates exists.

\section*{Acknowledgements}

Dorian Baudry thanks the support of the French National Research Agency: ANR-19-CHIA-02
SCAI, ANR-22-SRSE-0009 Ocean, and ANR-23-CE23-0002 Doom. Dorian Baudry was also partially funded by UK Research and Innovation (UKRI) under the UK government’s Horizon Europe funding guarantee [grant number EP/Y028333/1].

This research was supported in part by the French National Research Agency (ANR) in the framework of the PEPR IA FOUNDRY project (ANR-23-PEIA-0003) and through the grant DOOM ANR-23-CE23-0002. It was also funded by the European Union (ERC, Ocean, 101071601). Views and opinions expressed are however those of the author(s) only and do not necessarily reflect those of the European Union or the European Research Council Executive Agency. Neither the European Union nor the granting authority can be held responsible for them.



\printbibliography
\newpage
\appendix 
\section{Problem-specific simplifications} \label{appendix : }
This section focuses on characterizing undominated strategies and showing that when a learner plays on an $\epsilon$-discretization of the bid interval $[0,1]$ assuming ties never occur is without loss of generality.
\subsection{Dominated strategies} \label{appendix : dominated strategies}
Since the scale of the bid space is a deciding factor in the rates of regret we obtain, it is useful to analyze the utility functions of the learner. Indeed, for the learning procedure, it can be useful, to restrict ourselves from the start to bids which can potentially be optimal. We show next that, under certain condition on the values of the learner, we can restrict the bid space $B$. 

\begin{lemma} \label{lemma : dominated strategies}
Let $\{ v_1,v_2,...,v_k \}$ be the valuations of the learner. Then for any bids $\bvec=\{b_1,...,b_K\} \in B$ such that there exists $i\in [K], b_i >v_i $, there exists $\tilde{\mathbf{b}}$ a non-increasing sequence of $[0,v_1]\times [0,v_2] \times ... \times [0,v_K]$ such that : 
\begin{equation*}
    \forall \betavec \in B, u(\tilde{\mathbf{b}},\betavec) \geq u(\bvec,\betavec)
\end{equation*}
\end{lemma}

\begin{proof}
    Let $\{ v_1,v_2,...,v_k \}$ be the valuations of the learner. Let $\bvec=\{b_1,...,b_K\} \in B$ a bid such that there exists $i\in [K], b_i >v_i $. We define $\tilde{\mathbf{b}}$ as follows:
    \begin{equation*}
        \forall i \in [K], \tilde{\mathbf{b}}_i = \min (v_i,b_i)
    \end{equation*}
    Let $\betavec \in B$ be the bids of the adversary. As a max, $p(.)$ is increasing in its arguments hence \begin{equation} \label{eq : prices lower min value bid} p(\tilde{\mathbf{b}},\betavec) \leq p(\bvec,\betavec). \end{equation}

    There are two possible cases : 
    Either the allocation remains the same $x(\tilde{\mathbf{b}},\betavec) = x(\bvec,\betavec)$, or it decreases $x(\tilde{\mathbf{b}},\betavec) \leq x(\bvec,\betavec)$ (decreasing bids cannot result in a increased allocation).

    If the allocation remains the same, then \eqref{eq : prices lower min value bid} implies $u(\tilde{\mathbf{b}},\betavec) \geq u(\bvec,\betavec)$. 

    If the allocation decreases, the items obtained when bidding $\bvec$ but not obtained when bidding $\tilde{\mathbf{b}}$ necessarily corresponds to bids which have been lowered (other bids remain higher than the price). 
    Let $j \in [K]$ such that $b_j$ is one of these bids, since the item used to be won $b_j \geq p(\bvec,\betavec)$. Since it is not won by the learner playing $\tilde{\mathbf{b}}$, we have $\tilde{b}_j \leq p(\tilde{\mathbf{b}},\betavec)$. 
    Hence $v_j \leq p(\bvec,\betavec)$. 

    Since this is the case for all items $j$ won under $\bvec$ and not under bidding $p(\bvec,\betavec)$, $u(\tilde{\mathbf{b}},\betavec) \geq u(\bvec,\betavec)$.

    It is therefore always the case that \begin{equation}
        u(\tilde{\mathbf{b}},\betavec) \geq u(\bvec,\betavec)
    \end{equation}
    Since it is true for all $\betavec \in B$, this concludes the proof.
\end{proof}

A consequence of \autoref{lemma : dominated strategies} is that we can restrict our learning procedure to non-increasing sequence of $[0,v_1]\times [0,v_2] \times ... \times [0,v_K]$. 

In the following of the paper for simplicity we use $B$ as the bidding space. This covers the \emph{worst case} which corresponds to the case when for all $i \in [K]$, $v_i =1$. 

\subsection{Discretization error}

To use online learning techniques in our instance of multi-unit uniform price auction, we use a discretization $B_\epsilon$ of the bid space $B$ which is continuous. We bound here the added regret incurred because of this discretization, that is the additional regret suffered when comparing the best action in hindsight of $B_\epsilon$ to the best of $B$.

\begin{definition}[Discretized regret]
Let $\left ( \bvec^t\right )_{t \in [T]} \in B_\epsilon^T$  be the action played at time $t \in [T]$, against the opposing bids $\left ( \betavec^t\right )_{t \in [T]}\in B^T$. The discretized regret is defined as follows:
\begin{equation}
    R_{T,\epsilon}= \underset{\bvec \in B_\epsilon}{\max} \sum_{t=1}^T u(\bvec,\betavec^t)- \Expectation[\sum_{t=1}^T u(\bvec^t,\betavec^t)]\, .
\end{equation}
\end{definition}

We bound the cost of this discretization as follows :
\begin{lemma}\label{lemma : discretized regret}
    Let $\left ( \bvec^t\right )_{t \in [T]} \in B_\epsilon^T$  be the action played at time $t \in [T]$, against the opposing bids $\left ( \betavec^t\right )_{t \in [T]}\in B^T$. With $R_{ T,\epsilon}$ the dicretized regret and $R_T$ the regret, we have the following inequality: 
    \begin{align*}
        R_{T} \leq R_{T,\epsilon} + KT \epsilon \, .
    \end{align*}
\end{lemma}

\begin{proof}[Proof of Lemma \ref{lemma : discretized regret}]
Let $\left (\betavec^t\right )_{t\in [T]} \in \left ([0,1]^K \right ) ^T$ the adversary bids played up to time $T$. \newline
Let $(\bvec^t)_{t\in [T]} \in  B_\epsilon ^T$.

\begin{align*}
    R_T& = \underset{\bvec \in B}{\sup} \sum_{t=1}^T u(\bvec,\betavec^t)- \Expectation[\sum_{t=1}^T u(\bvec^t,\betavec^t)] \\
    & = \underset{\bvec \in B}{\sup} \sum_{t=1}^T u(\bvec,\betavec^t) -\underset{\bvec \in B_\epsilon}{\sup} \sum_{t=1}^T u(\bvec,\betavec^t) + \underset{\bvec \in B\epsilon}{\sup} \sum_{t=1}^T u(\bvec,\betavec^t)- \Expectation[\sum_{t=1}^T u(\bvec^t,\betavec^t)] \\ 
    & = \underset{\bvec \in B}{\sup} \sum_{t=1}^T u(\bvec,\betavec^t) -\underset{\bvec \in B_\epsilon}{\max} \sum_{t=1}^T u(\bvec,\betavec^t) + R_{T,\epsilon} \, .
\end{align*}

Let $\mu > 0$, there exists $\bvec_{opt} \in B$ such that $ \sum_{t=1}^T u(\bvec_{opt},\betavec^t) + \mu \geq \underset{\bvec \in B}{\sup} \sum_{t=1}^T u(\bvec,\betavec^t) $.

We define its closest discretized bid from above \begin{equation}
    \bvec_{opt,\epsilon} := \left ( \left \{ \begin{matrix*}  b_{opt,i}  &\text{if }  \frac{b_{opt,i}}{\epsilon} \in \mathbb{N}   \\
    1   &\text{if }  b_{opt,i} >  \left \lfloor \frac{1}{\epsilon} \right \rfloor \epsilon \\ 
    \left \lceil \frac{b_{opt,i}}{\epsilon} \right \rceil \epsilon
 & \text{else }\end{matrix*} \right . \right)_{i\in [K]} \ .
\end{equation}

Since $\max(.)$ cannot increase more than its arguments and $\forall i \in [K]$, 
 $b_{opt,\epsilon,i} \leq b_{opt,i} + \epsilon$, 
\begin{equation} \forall t \in [T], \  p(\bvec_{opt,\epsilon},\betavec^t)  \leq p(\bvec_{opt},\betavec^t) + \epsilon \, ,\end{equation}
and since $\forall i \in [K], \ b_{opt,\epsilon,i} \geq b_{opt,i}$,
\begin{equation}
    \forall t \in [T] , \ x( \bvec_{opt,\epsilon},\betavec^t)  \geq  x( \bvec_{opt},\betavec^t) \, 
    \end{equation}
therefore, \begin{equation*}
     \sum_{t=1}^T u(\bvec_{opt,\epsilon},\betavec^t) \geq  \sum_{t=1}^T u(\bvec_{opt},\betavec^t)- KT \epsilon \geq \underset{\bvec \in B}{\sup} \sum_{t=1}^T u(\bvec,\betavec^t) - KT\epsilon - \mu \, .
\end{equation*}

Hence
\begin{align*}
    R_T & = \underset{\bvec \in B}{\sup} \sum_{t=1}^T u(\bvec,\betavec^t) -\underset{\bvec \in B_\epsilon}{\max} \sum_{t=1}^T u(\bvec,\betavec^t) + R_{T,\epsilon} \\ 
     & \leq \underset{\bvec \in B}{\sup} \sum_{t=1}^T u(\bvec,\betavec^t) - \sum_{t=1}^T u(\bvec_{opt,\epsilon},\betavec^t) + R_{T,\epsilon} \\ 
    & \leq KT\epsilon + \mu + R_{T,\epsilon}.
\end{align*}

Since the previous inequality is true for any $\mu  > 0$, we get
\begin{equation*}
    R_T \leq KT\epsilon + R_{T,\epsilon}.
\end{equation*}

\end{proof}

\subsection{Avoiding ties} \label{appendix : avoiding ties}

In the latter analysis, we assume that ties never occur. We show here how this assumption, for our regret analysis, is equivalent to using a small perturbation of the bids of the learner. 
Let $\delta \in (0,\epsilon)$ and $X \sim \mathcal{U}[0,\delta]$ the random perturbation of the bids of the learner. We define $B_\epsilon^X$ the set of non-increasing sequences of $\left \{ X, \epsilon +X, 2\epsilon +X, ... , 1-\epsilon +X \right \}^K$, the set in which the perturbed bids of learners take value. 

This perturbation of the discretized set we use as bid space for the learner comes at no costs in terms of added regrets. The previous \autoref{lemma : discretized regret} can straightforwardly be applied to the perturbed set $B_\epsilon^X$, as the key reason for the additional regret is the discretization step $\epsilon$, which remains unchanged here.

\begin{lemma} \label{lemma : avoiding ties regret}
    Let $\betavec^T \in B^T$ be the bid of the adversary up to time $T$, for any bid sequence of the learner $\bvec^T \in B_\epsilon^{X\times T}$, there is almost surely never a tie. 
\end{lemma}

\begin{proof}
    Let $\betavec^T \in B^T$,and for all $t \in T$ denote $\betavec^t=\{\beta_1^t,...,\beta_K^t\}$.
    For any bid sequence $\bvec^T \in B_\epsilon^{X\times T}$, a necessary condition for ties to occur is that there exists $(t,j) \in [T]\times[K]$ such that $\beta_j^t \in \left  \{ \epsilon +X, 2\epsilon +X, ... , 1-\epsilon +X \right \}$. This is almost surely never the case, as it is the probability of $X$ to belong to a finite set. 
\end{proof}

\begin{remark} \label{rem : almost sure regret}
    For the sake of regret bounds, since we almost surely don't have any tie, the assumption that the adversary plays bids in $B_{\setminus \epsilon}$ is without loss of generality.
\end{remark}

\section{Appendix: Regret analysis} \label{app :learning}

\subsection{Full-information regret rates}

\begin{theorem} \label{theorem : full info regret}
    In the full information feedback setting, Algorithm \ref{algorithm : exp3 K bis} coincides with the Hedge algorithm with parameter $\eta$ on $H_\epsilon$. It ensures a regret of at most $O \left ( \sqrt{K^3T\log \left (T \right )} \right ) $ by taking $\epsilon =  \sqrt{\frac{K}{T}} $ and  $\eta = \sqrt{\frac{\log \left (\frac{T}{K} \right )}{2KT}}$.
\end{theorem}

\begin{proof}[Proof of Theorem \ref{theorem : full info regret}] \label{proof : th full inf}

In order to leverage classical results on the exponential weight algorithm in the expert setting, we aim to show that the combination of our algorithms leads to a probability update rule of the form: \begin{align} 
    \Proba^t(\hvec) & =  \frac{  \Proba^{t-1}(\hvec) \exp\left ({\eta u^{t}(\hvec)} \right ) }  {\sum_{\mathbf{j} \in H_\epsilon} \Proba^{t-1}(\mathbf{j}) \exp{\eta u^{t}(\mathbf{j})}} \, .
\end{align}

    Let $\hvec$ be bid profile in $H_\epsilon$, we denote $h_i$ its $i^{\text{th}}$ element ($i^{\text{th}}$ element of the sequence) and  $\text{len}(\hvec)$ the length of the sequence $\hvec$. Given the sampling Algorithm \ref{alg:weight_pushing}, we have by telescoping and the product of conditional probabilities 
    \begin{align*} \Proba^t(\hvec) &= \prod_{ i \in [1,\text{len}(\hvec)]} \Proba^t(h_i | h_{i-1} )  \\
    & = \prod_{ i \in [1,\text{len}(\hvec)]} \HalfProba[t](h_i ) \frac{\Gamma^t(h_i)}{\Gamma^t(h_{i-1})} \\
    &= \frac{\prod_{ i \in [1,\text{len}(\hvec)]} \HalfProba[t](h_i) }{\Gamma^t_0} \\
    & =  \frac{\prod_{ i \in [1,\text{len}(\hvec)]} \HalfProba[t-1](h_i) \exp \left ({\eta w^t(h_i)} \right )} {\Gamma^t_0} \\
    & = \frac{ \exp \left ({\eta u^t(\hvec)}\right ) }{\Gamma^t_0} \prod_{ i \in [1,\text{len}(\hvec)]} \Proba^{t-1}(h_i | h_{i-1}) \frac{\Gamma^{t-1}(h_{i-1})}{\Gamma^{t-1}(h_i)} \, .\\
    & \Proba^t(\hvec) = \frac{\Gamma^{t-1}_0}{\Gamma^t_0}  \exp\left ({\eta u^t(\hvec)} \right ) \Proba^{t-1}(\hvec)\, , \numberthis  \label{eq : induction probabilities}
    \end{align*}
where we used the probability update rule \eqref{eq : proba update rule}.

We can prove that $\Gamma^t_0 =\sum_{\hvec \in H_\epsilon} \prod\limits_{ i \in [1,\text{len}(\hvec)]} \HalfProba[t](h_i)$.
An induction on $k \in \left [\underset{\hvec \in H_\epsilon}{\max}  \text{len}(\hvec) \right ]$, where we denote $\hvec_{:k}$ the $k$ first component of the sequence $\hvec$: \begin{equation*}\Gamma^t_0 =\sum_{\hvec_{:k} : \hvec \in H_\epsilon} \left (\prod\limits_{ i \in [1,\min \{k,\text{len}(\hvec)\}]} \HalfProba[t](h_i) \right )\Gamma^t(h_{\min \{k,\text{len}(\hvec)\}})\end{equation*}
This is true for $k=1$ from the definition \eqref{def: gamma 0} of $\Gamma_0^t$. 

For $k \geq 1$,
\begin{align*}\Gamma^t_0 & =\sum_{\hvec_{:k} : \hvec \in H_\epsilon} \left (\prod\limits_{ i \in [1,\min \{k,\text{len}(\hvec)\}]} \HalfProba[t](h_i) \right )\Gamma^t(h_{\min \{k,\text{len}(\hvec)\}}) \\ 
& = \sum_{\hvec_{:k} : \hvec \in H_\epsilon} \left (\prod\limits_{ i \in [1,\min \{k,\text{len}(\hvec)\}]} \HalfProba[t](h_i) \right ) \sum_{h \in s(h_k)} \HalfProba[t](h) \Gamma^t(h) \\
& = \sum_{\hvec_{:k+1} : \hvec \in H_\epsilon} \left (\prod\limits_{ i \in [1,\min \{k+1,\text{len}(\hvec)\}]} \HalfProba[t](h_i) \right )\Gamma^t(h_{\min \{k+1,\text{len}(\hvec)\}})\, , \end{align*}
where we used the fact that the function successor $s(.)$ provides all possible next element of sequence $\hvec$. 

We then simplify $\frac{\Gamma^{t-1}_0}{\Gamma^t_0}$: \begin{align*} \Gamma^t_0& =\sum_{\hvec \in H_\epsilon} \prod\limits_{ i \in [1,\text{len}(\hvec)]} \HalfProba[t](h_i) \\
& = \sum_{\hvec \in H_\epsilon} \prod\limits_{ i \in [1,\text{len}(\hvec)]} \HalfProba[t-1](h_i) \exp \left ({\eta u^t(h_i)} \right ) \\
& = \sum_{\hvec \in H_\epsilon} \exp\left ({\eta u^t(\hvec)} \right ) \prod\limits_{ i \in [1,\text{len}(\hvec)]} \Proba_{t-1}(h_i | \hvec_{i-1}) \frac{\Gamma^{t-1}(\hvec_{i-1})}{\Gamma^{t-1}(h_i)} \\
& = \Gamma_0^{t-1} \sum_{\hvec \in H_\epsilon} \exp\left ({\eta u^t(\hvec)} \right ) \Proba_{t-1}(\hvec)\, . \end{align*}

We can therefore write \begin{align} 
    \Proba^t(\hvec) & =  \frac{  \Proba^{t-1}(\hvec) \exp\left ({\eta u^t(\hvec)} \right ) }  {\sum_{\mathbf{l} \in H_\epsilon} \Proba^{t-1}(\mathbf{l}) \exp{\eta u^t(\mathbf{l})}} \, .\label{eq : probability of exponential weight update} 
\end{align}

This is the update rule of the Hedge algorithm on the action space $H_\epsilon$. Therefore using Theorem 2.2 of \cite{cesa2006prediction}, restated in the Appendix~\ref{restated thm cesa} leads to the following regret bound: \begin{align*}
        R_{T,\epsilon} & \leq \frac{\log \left (|B_\epsilon|\right )}{\eta} + \frac{\eta T K^2}{8} \numberthis \label{eq:bound the Beps space} \\ 
    & \leq \frac{\log \left ( 1/ \epsilon^K\right )}{\eta} + \frac{\eta T K^2}{8}\, , \\ 
\end{align*}
where, to bound $\log \left (|B_\epsilon|\right )$ in \eqref{eq:bound the Beps space}, we used the fact that the action space $H_\epsilon$ is in bijection with the original discretized bid space : $K$ non-increasing elements of $ \{1,2, \hdots , \floor{\frac{1}{\epsilon}} \}$, which cardinal is trivially smaller than $\frac{1}{\epsilon^K} $.

Taking $\eta = \sqrt{\frac{\log \left ( \frac{1}{\epsilon} \right )}{KT}}$, we obtain \begin{equation*}
R_{T,\epsilon} \leq \frac{9}{8} \sqrt{T K^3 \log\left(\frac{1}{\epsilon} \right ) } \, .
\end{equation*}

Finally, using lemma \ref{lemma : discretized regret}, 
\begin{align*}
    R_T & \leq R_{T,\epsilon} +KT\epsilon \\
    & \leq \frac{9}{8} \sqrt{T K^3 \log\left(\frac{1}{\epsilon} \right )} + KT \epsilon \\
    & \leq \frac{9}{8} \sqrt{T K^3 } \left ( \log\left( \sqrt{\frac{T}{K}} \right ) +1 \right )\ , 
\end{align*}
with $\epsilon= \sqrt{\frac{K}{T}}$.
\end{proof}

\subsection{Partial feedback regret rates}

\subsubsection{Bandit feedback}
To prove the regret rates in the bandit feedback settings, we use Lemma~\ref{lemma : variance bandits estimator} which bounds the necessary quantity for the standard EXP3 analysis. 

\begin{lemma}[Bandit feedback estimator] \label{lemma : variance bandits estimator}
The estimator $\hat{u}^t(\hvec)$ defined by \eqref{def : subutility bandit estimator} has the following properties: 
\begin{itemize}
    \item The estimator $\hat{u}^t(\hvec)$ has a fixed bias $-K$: $$\Expectation \left [\hat{u}^t(\hvec) \right ] = u^t(\hvec) -K\, .$$
    \item The square of the estimator can be upper bounded as follows: $$\sum_{\hvec \in H_\epsilon}\Probat(\hvec) \Expectation [\hat{u}^t(\hvec)^2] \leq 4 K^2 \max \left (K^2,\frac{1}{\epsilon} \right )\, .$$
\end{itemize}
    
\end{lemma}

Since the estimator has a constant bias for every action, one can use it in the problem similarly to an unbiased estimator.

\begin{proof}[Proof of Lemma \ref{lemma : variance bandits estimator}] Let $h$ be a bid or a bid-gap, then
\begin{align*}
    \Expectation [\hat{u}^t(\hvec)] & =  \sum_{\hvec^t \in H_\epsilon} \Proba^t(\hvec^t)  \sum_{h \in \hvec} \hat{w}^t(h) \\
    &= \sum_{\hvec^t \in H_\epsilon} \Proba^t(\hvec^t)  \sum_{h \in \hvec} \mathbf{1}(h =  h^t_\star) \frac{w^t(h) - K }{\Proba^t(h)} \\
    &= \sum_{\hvec^t \in H_\epsilon} \Proba^t(\hvec^t)  \frac{w^t(h^t_\star(\hvec)) - K }{\Proba^t(h^t_\star(\hvec))} \mathbf{1}(h^t_\star(\hvec) =  h^t_\star) \\
    &= \frac{w^t(h^t_\star(\hvec)) - K }{\Proba^t(h^t_\star(\hvec))}  \sum_{\hvec^t \in H_\epsilon : h^t_\star(\hvec) \in \hvec^t} \Proba^t(\hvec^t)\\ 
    & = w^t(h^t_\star(\hvec)) - K  = u^t(\hvec) - K \, .
\end{align*}

\begin{align}
    \sum_{\hvec \in H_\epsilon} \Proba^t(\hvec) \Expectation [\hat{u}^t(\hvec)^2] &= \Expectation \left  [ \sum_{\hvec \in H_\epsilon} \Proba^t(\hvec)  \hat{u}^t(\hvec)^2 \right ] \\
    & = \sum_{\hvec^t \in H_\epsilon} \Proba^t(\hvec^t) \sum_{\hvec \in H_\epsilon}\Proba^t(\hvec)  \left (\frac{w^t(h^t_\star(\hvec)) - K }{\Proba^t(h^t_\star(\hvec))} \right )^2  \mathbf{1}(h^t_\star(\hvec) =  h^t_\star) \\
    &= \sum_{\hvec^t \in H_\epsilon} \Proba^t(\hvec^t)  \left (\frac{w^t(h^t_\star) - K }{\Proba^t(h^t_\star)} \right )^2  \sum_{\hvec \in H_\epsilon : h^t_\star \in \hvec} \Proba^t(\hvec) \\
    & = \sum_{\hvec^t \in H_\epsilon} \Proba^t(\hvec^t)  \frac{ \left ( w^t(h^t_\star) - K  \right )^2}{\Proba^t(h^t_\star)}   \\
    & \leq K^2 \sum_{h^t_\star \in \Ou^t}  \frac{\sum_{\hvec^t \in H_\epsilon : h^t_\star \in \hvec^t} \Proba^t(\hvec^t)}{\Proba^t(h^t_\star)} \\
    & \leq K^2 \sum_{h^t_\star \in \Ou^t} 1 \\
    & \leq (K)^2 |\Ou^t| \, ,
\end{align}
where $$\Ou^t= \left \{ h^t_\star(\hvec) \bigg{|} \hvec \in H_\epsilon \right  \}\, .$$

There only remains to upper bound $\card{\Ou^t}$. 

For any $p \in [0,1]$ we denote for this proof $\smash{j(p)= \floor{\frac{p}{\epsilon}}}$, which is the value $j$ such that $p \in [j\epsilon , j \epsilon + \epsilon )$. 
Let $\hvec \in H_\epsilon$, notice that $h_\star^t(\hvec,\betavec^t) \in \{ b_{x_t(\hvec,\betavec^t),j\left (p(\hvec,\betavec^t)\right)}, \ b_{x_t(\hvec,\betavec^t)+\frac{1}{2},j\left (p(\hvec,\betavec^t)\right)} \}$  which directly results from the decomposition formula \ref{eq : decompositon of utility a.s.}. To upper bound $\card{\Ou^t}$ we will therefore upper bound the different values the pair  $\left (x_t(\ . \ ,\betavec^t),p(\ . \ ,\betavec^t)\right )$ can take.\newline 

Since the learner plays bids in $H_\epsilon$ (or equivalently $B_\epsilon$), $p(.,\betavec^t)$ can only take either the value of one of the components in $\betavec$ or one of the bids of its first argument. \newline 
Because $\betavec^t$ is a vector of size K we can write that: $\card{\{ p(\hvec,\betavec^t), \hvec \in H_\epsilon \}} \leq K + \floor{\frac{1}{\epsilon}}$. 

Furthermore, because units are either attributed to the player or the adversary, we can write 
$ K - \card{ \{\beta_i^t > p_t(\ . \ ,\betavec^t) \} } \geq x_t(\ . \ ,\betavec^t) \geq  K - \card{ \{\beta_i^t \geq p_t(\ . \ ,\betavec^t) \} }$. The two cardinals can only differ if the price is set by an adversary bid because the no ties assumption implies almost surely for all $ i \in [K], \beta_i^t \notin \left [ \frac{1}{\epsilon} \right ] $.

Therefore each possible value of $p_t(\ . \ ,\betavec^t)$ only correspond to one value of $x_t(\ . \ ,\betavec^t)$, except for the $K$ values set by the adversary, where $x_t(\ . \ ,\betavec^t)$ can at most take $K$ values.

Therefore 
\begin{align}
   2 \left ( K^2 + \left \lfloor \frac{1}{\epsilon} \right \rfloor \right ) &\geq 2 \card{ \{ \left (x_t(\hvec,\betavec^t),p(\hvec ,\betavec^t)\right ) , \hvec \in H_\epsilon\}} \\
    & \geq 2 \card{ \{ \left (x_t(\hvec,\betavec^t),j\left (p(\hvec ,\betavec^t)\right )\right ) , \hvec \in H_\epsilon \} } \\ 
    & \geq \card{\Ou^t}\, ,
\end{align}
which leads to the needed upper bounds

\end{proof}

We now restate Theorem~\ref{theorem : Regret exp3} and then provide proof of the corresponding regret guarantees.

\begingroup
\def\thetheorem{~\ref{theorem : Regret exp3}}
\begin{theorem}
   In the repeated $K$-unit auction with uniform pricing guarantees and under bandit feedback, Algorithm \ref{algorithm : exp3 K bis}  incurs a regret of at most $\mathcal{O} \left ( K^{4/3} T^{2/3} \log \left (T \right ) \right ) $. For any time horizon $T$, with the choices of $\epsilon = \left ( \frac{K}{T} \right )^{1/3}$ and  $\eta = K^{-1/3}T^{-2/3}\sqrt{\log \left ( \frac{T}{K} \right )/3}$.
\end{theorem}
\addtocounter{theorem}{-1}
\endgroup
 
\begin{proof}[Proof of Theorem \ref{theorem : Regret exp3}] \label{proof : theorem exp3}

For $T \in  \mathbf{N}$, we denote $\left (\hvec^t\right )_{t\in [T]} \in H_\epsilon^T$ the actions played at each time-steps, generated by Algorithm \ref{algorithm : exp3 K bis}. \newline 
We can first notice that, by conducting the same analysis as in \ref{proof : th full inf} up to Equation \eqref{eq : probability of exponential weight update}, we  obtain: 

\begin{align} 
    \Proba^t(\hvec) & =  \frac{  \Proba_{t-1}(\hvec) \exp\left ({\eta \hat{u}^t(\hvec)} \right ) }  {\sum_{\mathbf{l} \in H_\epsilon} \Proba_{t-1}(\mathbf{l}) \exp{\eta \hat{u}^t(\mathbf{l})}}\ ,\label{eq : probability of exponential weight update bandit} 
\end{align}
which by simple induction allows us to obtain:

\begin{align} 
    \Proba^t(\hvec) & =  \frac{ \exp\left ( \sum_{j=1}^t {\eta \hat{u}^j(\hvec)} \right ) }  {\sum_{\mathbf{l} \in H_\epsilon} \exp \left (\sum_{j=1}^t  \eta \hat{u}^j(\mathbf{l})\right )} \ . \label{eq : probability of exponential weight update bandit recur} 
\end{align}

We can now proceed to the regret analysis. \newline
For any action $\hvec \in H_\epsilon$, we define : 
$$R_{T,\hvec}=\sum_{t=1}^{T} u^t(\hvec) -\Expectation \left [\sum_{t=1}^{T} u^t(\hvec^t)\right ]\, , $$
which is the expected regret relative to playing $\hvec$ in all the rounds. 

We have, because of Lemma \ref{lemma : variance bandits estimator}, 
$\Expectation \left [ \sum_{t=1}^{T} \hat{u}^t(\hvec) \right ] = \sum_{t=1}^{T} u^t(\hvec^t) - KT$ \newline 
and $\Expectation_{t-1} \left [ u^t(\hvec^t)\right ] = \sum_{\hvec \in B} \Proba^t(\hvec) u^t(\hvec) = \sum_{\hvec \in B} \Proba^t(\hvec) \Expectation_{t-1} \left [ \hat{u}^t(\hvec)\right ] + K$.

Therefore 
\begin{equation} \label{equation : regret ,1}
R_{T,\hvec}=\Expectation \left [ \sum_{t=1}^{T} \hat{u}^t(\hvec) \right ] -\Expectation \left [\sum_{t=1}^T \sum_{\hvec \in H_\epsilon} \Proba^t(\hvec) \hat{u}^t(\hvec)\right ] \ .\end{equation}

We denote $W_n= \sum_{\hvec \in H_\epsilon} \exp ( \eta  \sum_{t=1}^{n} \hat{u}^t(\hvec) ) $.

Then we have for any $\hvec \in H_\epsilon$, $$ \exp \left  ( \eta  \sum_{t=1}^{T} \hat{u}^t(\hvec) \right )  \leq  \sum_{\hvec \in H_\epsilon} \exp ( \eta  \sum_{t=1}^{T} \hat{u}^t(\hvec) ) = W_T= W_0 \prod_{t=1}^T \frac{W_t}{W_{t-1}}\ . $$
We can then upper bound the terms of the product as follows :
\begin{align*} \frac{W_t}{W_{t-1}} & \leq \sum_{\hvec \in H_\epsilon} \frac{\exp \left ( \eta  \sum\limits_{l=1}^{t-1} \hat{u}^l(\hvec)\right )}{W_{t-1}} \exp \left (\eta \hat{u}^t(\hvec) \right ) \\
&\leq  \sum_{\hvec \in H_\epsilon} \Proba^t(\hvec) \exp \left (\eta \hat{u}^t(\hvec) \right ) \ ,\end{align*}
where the second inequality comes from \ref{eq : probability of exponential weight update bandit recur}. 

We can then further bound this term using the inequalities $$\forall \  x\leq1, \  \exp(x)\leq 1+x+x^2 \text{ and }\forall x \in\mathbb{R}, 1+x\leq \exp(x)\, .$$

This gives \begin{align}
\frac{W_t}{W_{t-1}} &\leq 1 + \eta \sum_{\hvec\in H_\epsilon} \Proba^t(\hvec)\hat{u}^t(\hvec) + \eta^2 \sum_{\hvec\in H_\epsilon} \Proba^t(\hvec) \hat{u}^t(\hvec)^2\\
& \leq \exp\left (  \eta \sum_{\hvec\in H_\epsilon} \Proba^t(\hvec) \hat{u}^t(\hvec) + \eta^2 \sum_{\hvec\in H_\epsilon} \Proba^t(\hvec)\hat{u}^t(\hvec)^2 \right ) \ .
\end{align}

This in turn yields
\begin{align}
    \exp \left ( \eta  \sum_{t=1}^{T} \hat{u}^t(\hvec) \right  ) & \leq W_0 \prod_{t=1}^T \exp\left (  \eta \sum_{\hvec\in H_\epsilon} \Proba^t(\hvec)\hat{u}^t(\hvec) + \eta^2 \sum_{\hvec\in H_\epsilon} \Proba^t(\hvec) \hat{u}^t(\hvec)^2 \right )  \\
    &  \leq W_0 \exp \left (  \eta  \sum_{t=1}^T \sum_{\hvec\in H_\epsilon} \Proba^t(\hvec)\hat{u}^t(\hvec) + \eta^2 \sum_{t=1}^T \sum_{\hvec\in H_\epsilon} \Proba^t(\hvec) \hat{u}^t(\hvec)^2 \right )\ ,\end{align}
where by applying the log, simplifying, and taking the expectation we get 
\begin{align}
    \sum_{t=1}^{T} \hat{u}^t(\hvec) -  \sum_{t=1}^T \sum_{\hvec\in B} \Proba^t(\hvec)\hat{u}^t(\hvec) & \leq \frac{\log( W_0 )}{\eta} + \eta \sum_{t=1}^T \sum_{\hvec\in H_\epsilon} \Proba^t(\hvec) \hat{u}^t(\hvec)^2 \\
    \Expectation\left [  \sum_{t=1}^{T} \hat{u}^t(\hvec) -  \sum_{t=1}^T \sum_{\hvec\in H_\epsilon} \Proba^t(\hvec) \hat{u}^t(\hvec)\right ] & \leq \frac{\log( W_0 )}{\eta} + \eta  \Expectation \left [ \sum_{t=1}^T \sum_{\hvec\in H_\epsilon} \Proba^t(\hvec) \hat{u}^t(\hvec)^2\right ] \\
    R_{T,\hvec} & \leq \frac{\log( W_0 )}{\eta} + \eta  \Expectation \left [ \sum_{t=1}^T  \sum_{\hvec\in B} \Proba^t(\hvec) \hat{u}^t(\hvec)^2\right ]\ .
\end{align}
We recognize the expression of the regret from \eqref{equation : regret ,1}. Because it is true for all $\hvec \in H_\epsilon$, we can take the maximum and notice : $R_{T,\epsilon} = \max_{\hvec \in H_\epsilon} R_{T,\hvec}$. 
Noticing that $W_0 = |H_\epsilon| \leq \left ( \frac{1}{\epsilon} \right)^K$ and using Lemma \ref{lemma : variance bandits estimator}  concludes the bound on the discritzed regret as follows:

\begin{align}
    R_{T,\epsilon} & \leq \frac{\log{| H_\epsilon|}}{\eta} +  \eta  \sum_{t=1}^{T} \sum_{\hvec \in H_\epsilon} \Proba^t(\hvec) \Expectation [\hat{u}^t(\hvec)^2] \label{eq : bound discretized regret bandits} \\
    & \leq K \frac{\log \left ( \frac{1}{\epsilon} \right )}{\eta} + \eta 4 K^2T\frac{1}{\epsilon} \\
    & \leq 5 K \sqrt{\frac{KT}{\epsilon} \log \left ( \frac{1}{\epsilon} \right )}\ ,
\end{align}
with $\eta = \sqrt{\frac{ \epsilon }{KT} \log \left ( \frac{1}{\epsilon} \right ) } $

Then using Lemma \ref{lemma : discretized regret}, we can bound the regret: 
\begin{align}
    R_T& = R_{T,disc} + KT\epsilon \\
    & \leq 5 K^{3/2} \sqrt{\frac{T}{\epsilon} \log \left ( \frac{1}{\epsilon} \right )} + KT\epsilon \\ 
    & \leq 5 K^{4/3} T^{2/3}  \left ( 1 + \frac{1}{3}\log \left ( \frac{T}{K} \right ) \right )  \ ,
\end{align}
with the specific choice of  $\epsilon =   \left ( \frac{K}{T} \right )^{1/3}$.
\end{proof}

\subsubsection{All-winner feedback}

As in the bandit feedback, to prove the regret rates, we use Lemma~\ref{lemma : variance all-winner estimator} which bounds the necessary quantity for the standard EXP3 analysis.

\begin{lemma}\label{lemma : variance all-winner estimator}
The estimator $\bar{u}^t(\hvec)$, defined by \eqref{def : subutilities all winner estimator } has the following properties: 
\begin{itemize}
    \item The estimator $\bar{u}^t(\hvec)$ has a fixed bias $-K$: \begin{equation}
        \Expectation \left [\bar{u}^t(\hvec) \right ] = u^t(\hvec) -K \, .\end{equation}
    \item The square of the estimator verifies: \begin{equation} \sum_{\hvec \in H_\epsilon} \Probat(\hvec) \Expectation [\bar{u}^t(\hvec)^2] \leq 8 K^4 \log (2)\, .\end{equation}
\end{itemize}
\end{lemma}

\begin{proof}[Proof of Lemma \ref{lemma : variance all-winner estimator}] This proof is mostly based on the following careful computations.
\begin{align*}
    \Expectation [\bar{u}^t(\hvec)] & =  \sum_{\hvec^t \in B} \Proba^t(\hvec^t)  \sum_{h \in \hvec} \bar{w}^t(h) \\
    &= \sum_{\hvec^t \in H_\epsilon} \Proba^t(\hvec^t)  \sum_{h_{k,j} \in \hvec} \mathbf{1} \left ( h \in \Observed^t_\star( \hvec^t) \right ) \frac{w^t(h)-K}{\underset{\mathbf{l}^t \sim \mathcal{B}^t}{\Proba}(h \in \Observed^t_\star (\mathbf{l}^t) )} \numberthis\\
    &= \sum_{\hvec^t \in H_\epsilon} \Proba^t(\hvec^t)  \frac{w^t(h^t_\star(\hvec)) - K }{\underset{\mathbf{l}^t \sim \mathcal{B}^t}{\Proba}(h^t_\star(\hvec) \in \Observed^t_\star (\mathbf{l}^t) )} \mathbf{1}(h^t_\star(\hvec) \in   \Observed^t_\star (\hvec^t)) \\
    &= \frac{w^t(h^t_\star(\hvec)) - K }{\underset{\mathbf{l}^t \sim \mathcal{B}^t}{\Proba}\left (h^t_\star(\hvec) \in \Observed^t_\star (\mathbf{l}^t) \right )}  \sum_{\hvec^t \in H_\epsilon : h^t_\star(\hvec) \in \Observed^t_\star (\hvec^t)} \Proba^t(\hvec^t)\\ 
    & = w^t(h^t_\star(\hvec)) - K  = u^t(\hvec) - K\, .
\end{align*}

\begin{align*}
    \sum_{\hvec \in H_\epsilon} \Proba^t(\hvec) \Expectation [\hat{u}^t(\hvec)^2] &= \Expectation \left  [ \sum_{\hvec \in H_\epsilon} \Proba^t(\hvec)  \hat{u}^t(\hvec)^2 \right ] \\
    & = \sum_{\hvec \in H_\epsilon} \Proba^t(\hvec) \sum_{\hvec^t \in H_\epsilon}\Proba^t(\hvec^t)  \left (\frac{u^t(h^t_\star(\hvec)) - K }{\underset{\mathbf{l}^t \sim \mathcal{B}^t}{\Proba}(h^t_\star(\hvec) \in \Observed^t_\star (\mathbf{l}^t) )} \right )^2  \mathbf{1}(h^t_\star(\hvec) \in   \Observed^t_\star (\hvec^t)) \\
    &= \sum_{\hvec \in H_\epsilon} \Proba^t(\hvec)  \left (\frac{u^t(h^t_\star(\hvec)) - K }{\underset{\mathbf{l}^t \sim \mathcal{B}^t}{\Proba}(h^t_\star(\hvec) \in \Observed^t_\star (\mathbf{l}^t) )} \right )^2  \sum_{\hvec^t \in H_\epsilon : h^t_\star(\hvec) \in \Observed^t_\star (\hvec^t)} \Proba^t(\hvec^t) \\
    & = \sum_{\hvec \in H_\epsilon} \Proba^t(\hvec)  \frac{ \left ( u^t(h^t_\star(\hvec)) - K  \right )^2}{\underset{\mathbf{l}^t \sim \mathcal{B}^t}{\Proba}(h^t_\star(\hvec) \in \Observed^t_\star (\mathbf{l}^t) )}   \\
    & \leq K^2 \sum_{\hvec \in H_\epsilon} \frac{\Proba^t(\hvec)}{\underset{\mathbf{l}^t \sim \mathcal{B}^t}{\Proba}(h^t_\star(\hvec) \in \Observed^t_\star (\mathbf{l}^t) )} \\
    & \leq  K^2 \sum_{h^t_\star \in \Ou^t} \frac{\Proba^t(h^t_\star)}{\underset{\mathbf{l}^t \sim \mathcal{B}^t}{\Proba}(h^t_\star \in \Observed^t_\star (\mathbf{l}^t) )} \\
    & \leq  K^2 \sum_{h^t_\star \in \Ou^t} \frac{\Proba^t(h^t_\star)}{\sum_{a \in \Ou^t : h^t_\star \in \Observed^t_\star(a) } \Proba^t(a)} \numberthis \label{eq : where we use the graph stuff}\\ 
    & \leq K^2  8K \log \left ( 2 \frac{1}{\epsilon^{2K} \alpha K } \right ) \\
    & \leq 8K^4   \log \left ( \frac{2}{\epsilon}  \right ) \ .\\
\end{align*}

Taking $\alpha = \frac{1}{K}$.

Where to bound \eqref{eq : where we use the graph stuff}, we use lemma \ref{lemma : directed_graph} from \cite{Alon_Graph-Structured_Feedback}, restated in the Appendix. 

We define a graph over the elements of $\mathcal{O}^t$, such that each element $o_1$ has an incoming edge from the other elements $o_2$ such that $o_1 \in \Observed^t_\star(o_2)$. This graph matches  \eqref{eq : where we use the graph stuff} to the expression lemma \ref{lemma : directed_graph} allows to bound.

It only remains to determine the independence number of this graph. First notice that, for each value of $k+\frac{1}{2}$ only one bid-gaps with this first index can belong to $\mathcal{O}^t$. Indeed, otherwise, since there exists a bid-profile $\hvec$ such that both belong to it, $h_\star^t(\hvec)$ would have two values, which is impossible because only one bid or bid-gap per bid profile can have non-zero sub-utility. 

Then notice that for two bids in $\mathcal{O}^t$, with the same first index $k$ an integer values, the observed set of the \emph{lowest} one necessarily contains the other. This naturally arises from the definition of $\Observed^t$. 

These two observations ensure that, in an independent set of this graph, there is at most one element having each index $k \in \{1,\frac{3}{2},2,...\frac{2K-1}{2K}, K \}$.
This ensures the independence number of this is at most $2K$. Which using the lemma \ref{lemma : directed_graph} from \cite{Alon_Graph-Structured_Feedback}, completes the proof.

\end{proof}

We restate the regret guarantees in the all-winner feedback before the proof.
\begingroup
\def\thetheorem{~\ref{theorem : Regret all winner}}
\begin{theorem}
   For any time horizon $T$, using Algorithm \ref{algorithm : exp3 K bis} in the repeated $K$-unit auction with uniform pricing guarantees, under all-winner feedback, a regret of at most $\smash{\mathcal{O} \left ( K^{5/2} \sqrt{T} \log (T) \right ) }$ with $\eta =K^{-1}T^{-1/2}$ and  $\epsilon = K^{3/2}T^{1/2}$.
\end{theorem}
\addtocounter{theorem}{-1}
\endgroup

\begin{proof}[Proof of Theorem \ref{theorem : Regret all winner}] \label{proof : theorem all winner}
    The proof of this theorem is identical to the one of theorem \ref{theorem : Regret exp3}, with only the need to replace $\hat{u}^t$ by $\bar{u}^t$ up to the point where we bound the regret in equation \ref{eq : bound discretized regret bandits}. 
The proof completes as follows.
The discretized regret can be bounded as in \ref{proof : theorem exp3}:
\begin{align}
    R_{T,disc} & \leq \frac{\log{| H_\epsilon|}}{\eta} +  \eta  \sum_{t=1}^{T} \sum_{\hvec \in H_\epsilon} \Proba^t(\hvec) \Expectation [\bar{u}^t(\hvec)^2] \\
    & \leq K \frac{\log \left ( \frac{1}{\epsilon} \right )}{\eta} + \eta 8 K^4T \log\left(\frac{2}{\epsilon}\right) \\
    & \leq   K^{5/2} \sqrt{T} \left ( 8 \log (2) + 9 \log \left ( \frac{1}{\epsilon} \right ) \right )\ ,
\end{align}
with $\eta = \frac{1}{K\sqrt{T}}$.

Then using Lemma \ref{lemma : discretized regret}, we can bound the regret: 
\begin{align}
    R_T& = R_{T,disc} + KT\epsilon \\
    & \leq K^{5/2} \sqrt{T} \left ( 8\log(2) + 9\log \left ( \frac{1}{\epsilon} \right ) \right ) + KT\epsilon \\ 
    & \leq  K^{5/2} \sqrt{T} \left ( 1 + 8\log(2) + \frac{9}{2}\log \left ( \frac{T}{K^3} \right ) \right )\ ,
\end{align}
with $\epsilon =   \sqrt{\frac{K^3}{T}}$
\end{proof}

\section{Proof of technical lemmas} \label{proof}

\observedAllWinner*

\begin{proof}[Proof of Lemma \ref{lemma : observed all-winner}]
Let $t \in [T], \hvec^t \text{ and } \betavec^t$ be the action of the player and the adversary at time $t$. Let $\bvec^t$ be the corresponding bid to the pseudo-bid $\hvec^t$.
Under the all-winner feedback, all winning bids are revealed, hence the feedback reveals to the learner the $K-x(\bvec^t,\betavec^t)$ biggest bids of the adversary : $(\beta_i)_{i\leq K-x(\bvec^t,\betavec^t)}$.
Furthermore, since the price  is known, the learner can deduce from the rules of the auction that for all $i \geq K-x(\bvec^t,\betavec^t)$ : \begin{equation} \label{equation : observed all winner price} \beta_i \leq p(\bvec^t,\betavec^t). \end{equation}

For any value $k,j \in \mathcal{K} \times \mathcal{J}_\epsilon$ such that $h_{k,j} \in \Observed(\hvec,\betavec)$, let's show that we can evaluate the corresponding sub-utilities.

We first look at the ability of the learner to evaluate the indicator functions in the sub-utilities defined in \autoref{lemma : decomposition}, for $h_{k,j} \in \Observed(\hvec^t,\betavec^t)$. 

For the integer values of $k$, we can rewrite the indicator function of the sub-utilities, as follows :
$\indicator \{ p_H(\hvec,\betavec) = j\epsilon \} \cap \{ x_H(\hvec,\betavec) = k \} = \indicator \{ \beta_{K-k} >  j\epsilon > \beta_{K-k+1} \} $.
When $K-k+1 \leq K-x(\bvec^t,\betavec^t)$ this can be evaluated for any value of $j$, because the adversary bids are known. 
When $K-k = K-x(\bvec^t,\betavec^t)$, the indicator function can still be evaluated if $j\epsilon > p(\bvec^t,\betavec^t)$ using \eqref{equation : observed all winner price}.

For the half-integer values of $k$, we can rewrite the indicator function of the sub-utilities \autoref{lemma : regularity over outcomes} as follows :
$\indicator \{ p_H(\hvec,\betavec) \in (j\epsilon,(j+1) \epsilon) \} \cap \{ x_H(\hvec,\betavec) = k-1/2 \} = \indicator \{j\epsilon< \beta_{K-k+1/2} <  (j+1)\epsilon \} $.
Therefore, when $K-k+1/2 \leq K-x(\bvec^t,\betavec^t)$ this indicator function can be evaluated. Hence when $k \geq x+1/2$, and that regardless of the value of $j$.

Therefore, it is always possible for the learner to evaluate the indicator function. 

Evaluating the remaining term of the sub-utilities $\sum_{l=1}^{\floor{k}} v_l - p_H(\hvec,\betavec)$ is more straightforward since it only needs to be done when the indicator function takes value 1. 

For the integer values of $k$, if the indicator function takes value 1, then $p_H(\hvec,\betavec) = j\epsilon$, therefore the remaining term is known. 

For the half-integer values of $k$, if the transformed indicator function takes value 1, then the price is set by $\beta_{K-k+1/2}$, therefore, the remaining term is also known. 

This concludes the proof as the full sub-utilities can always be evaluated on $\Observed(\hvec,\betavec)$.
\end{proof}

\section{Restated results from the literature}

\subsection{Exponential weight forecaster}
In this problem of learning under expert advice, they are $N$ expert and at each time $t \in [N]$, the learner chooses a probability to play each expert $ (y^t_i)_{i\in[N]} \in \mathcal{Y}$ and nature reveals the losses $(l^t_i)_{i\in[N]} \in [0,L]^N$.

\begin{equation*}
    R_n = \sum_{t=1}^n \sum_{i=1}^N y^t_i l^t_i - \min_{i \in [N] } \left ( \sum_{t=1}^n l^t_i \right )\ .
\end{equation*}

\begin{theorem} [Theorem 2.2 \cite{cesa2006prediction}] \label{restated thm cesa}
Assume that the losses $l$ take values in $[0,L]$. For any $n$ and $\eta>0$, and for all $y_1, \ldots, y_n \in \mathcal{Y}$, the regret of the exponentially weighted average forecaster satisfies
$$
R_n \leq \frac{\log N}{\eta}+\frac{n L^2 \eta}{8} .
$$

In particular, with $\eta=\sqrt{8 \ln N / n}$, the upper bound becomes $\sqrt{(n / 2) \ln N}$.

\end{theorem} 
This theorem, besides the changes in notations, is a slight variation from the original formulation as it allows for losses greater than 1. The resulting $L^2$ term in the upper bound is a well known extension and the steps to prove this extension to scaled losses are provided in the original work by \cite{cesa2006prediction}.

\subsection{Lemma graph feedback}

The following lemma is restated from \cite{Alon_Graph-Structured_Feedback}.

\begin{lemma} \label{lemma : directed_graph}
Let $G = (V,E)$ be a directed graph with $\abs{V}=K$, in which each node $ i \in V$ is assigned a positive weight $w_{i}$.
Assume that $\sum_{i \in V} w_{i} \leq 1$, and that $w_{i} \geq \epsilon$ for all $i \in V$ for some constant $0<\epsilon<\frac{1}{2}$.
Then
\begin{align*}
	\sum_{i \in V} \frac{w_{i}}{w_{i}+\sum_{j \in \nin(i)} w_{j}}
	\leq 4 \alpha \ln \frac{4K}{\alpha\epsilon}\ ,
\end{align*}
where $\alpha = \alpha(G)$ is the independence number of $G$.
\end{lemma}

\subsection{First price auction lower bound}

We restate Theorem 10 from \citep{balseiro2019contextual} :

\begin{theorem}[Lower Bound for Learning to Bid] Any algorithm must incur $\Omega(T^{2/3})$ regret for the learning to bid in first-price auctions problem, even if the value of the bidder is fixed
(i.e., there is only one context).
\end{theorem}

\newpage
\section*{NeurIPS Paper Checklist}

\begin{enumerate}

\item {\bf Claims}
    \item[] Question: Do the main claims made in the abstract and introduction accurately reflect the paper's contributions and scope?
    \item[] Answer: \answerYes{}
    \item[] Justification: The main claims made in the abstract are presented in the paper, specifically as the two main Theorem~\ref{theorem : Regret exp3} and Theorem~\ref{theorem : Regret all winner} and Lemma~\ref{lemma : lower bound}.

\item {\bf Limitations}
    \item[] Question: Does the paper discuss the limitations of the work performed by the authors?
    \item[] Answer: \answerYes{}
    \item[] Justification: The paper discuss and introduces the assumptions made in order for the stated results to hold, most of which are stated in the Introduction~\ref{introduction}.

\item {\bf Theory Assumptions and Proofs}
    \item[] Question: For each theoretical result, does the paper provide the full set of assumptions and a complete (and correct) proof?
    \item[] Answer: \answerYes{} 
    \item[] Justification: While the full set of assumption is presented as part of the problem setting in the Introduction~\ref{introduction}, proofs of the theorems are provided in the appendix ~\ref{app :learning}.

    \item {\bf Experimental Result Reproducibility}
    \item[] Question: Does the paper fully disclose all the information needed to reproduce the main experimental results of the paper to the extent that it affects the main claims and/or conclusions of the paper (regardless of whether the code and data are provided or not)?
    \item[] Answer: \answerNA{}. 
    \item[] Justification: The paper does not include experiments.

\item {\bf Open access to data and code}
    \item[] Question: Does the paper provide open access to the data and code, with sufficient instructions to faithfully reproduce the main experimental results, as described in supplemental material?
    \item[] Answer:\answerNA{} 
    \item[] Justification: The paper does not include experiments.

\item {\bf Experimental Setting/Details}
    \item[] Question: Does the paper specify all the training and test details (e.g., data splits, hyperparameters, how they were chosen, type of optimizer, etc.) necessary to understand the results?
    \item[] Answer: \answerNA{}. 
    \item[] Justification: The paper does not include experiments.

\item {\bf Experiment Statistical Significance}
    \item[] Question: Does the paper report error bars suitably and correctly defined or other appropriate information about the statistical significance of the experiments?
    \item[] Answer: \answerNA{}. 
    \item[] Justification: The paper does not include experiments.

\item {\bf Experiments Compute Resources}
    \item[] Question: For each experiment, does the paper provide sufficient information on the computer resources (type of compute workers, memory, time of execution) needed to reproduce the experiments?
    \item[] Answer: \answerNA{}. 
    \item[] Justification: The paper does not include experiments.
    
\item {\bf Code Of Ethics}
    \item[] Question: Does the research conducted in the paper conform, in every respect, with the NeurIPS Code of Ethics \url{https://neurips.cc/public/EthicsGuidelines}?
    \item[] Answer: \answerYes{} 
    \item[] Justification: 

\item {\bf Broader Impacts}
    \item[] Question: Does the paper discuss both potential positive societal impacts and negative societal impacts of the work performed?
    \item[] Answer: \answerNA{} 
    \item[] Justification: This is a theoretical paper, its result are not tied to a specific field. While it might be the basis for further research into applying learning in auction, which would allow for participant in auction to better adapt to others strategies, it is unclear what societal impact this might have and how fit for practical use the techniques develloped here are.
    
\item {\bf Safeguards}
    \item[] Question: Does the paper describe safeguards that have been put in place for responsible release of data or models that have a high risk for misuse (e.g., pretrained language models, image generators, or scraped datasets)?
    \item[] Answer: \answerNA{} 
    \item[] Justification: The paper poses no such risks.

\item {\bf Licenses for existing assets}
    \item[] Question: Are the creators or original owners of assets (e.g., code, data, models), used in the paper, properly credited and are the license and terms of use explicitly mentioned and properly respected?
    \item[] Answer: \answerNA{} 
    \item[] Justification: The paper does not use existing assets.

\item {\bf New Assets}
    \item[] Question: Are new assets introduced in the paper well documented and is the documentation provided alongside the assets?
    \item[] Answer: \answerNA{} 
    \item[] Justification: The paper does not release new assets.

\item {\bf Crowdsourcing and Research with Human Subjects}
    \item[] Question: For crowdsourcing experiments and research with human subjects, does the paper include the full text of instructions given to participants and screenshots, if applicable, as well as details about compensation (if any)? 
    \item[] Answer: \answerNA{} 
    \item[] Justification: The paper does not involve crowdsourcing nor research with human subjects.

\item {\bf Institutional Review Board (IRB) Approvals or Equivalent for Research with Human Subjects}
    \item[] Question: Does the paper describe potential risks incurred by study participants, whether such risks were disclosed to the subjects, and whether Institutional Review Board (IRB) approvals (or an equivalent approval/review based on the requirements of your country or institution) were obtained?
    \item[] Answer: \answerNA{} 
    \item[] Justification:  The paper does not involve crowdsourcing nor research with human subjects.

\end{enumerate}

\end{document}